\documentclass[a4paper,12pt]{article}
\usepackage{amsmath,amsthm,amsfonts,amssymb,bbm}
\usepackage{graphicx,psfrag,subfigure,color,cite}
\usepackage{tikz-cd} 
\usepackage{tikz}
\numberwithin{equation}{section}

\newcommand{\e}{\varepsilon}
\newcommand{\Pb}{\mathbb{P}}
\newcommand{\E}{\mathbb{E}}
\newcommand{\I}{\mathrm{i}}

\newcommand{\R}{\mathbb{R}}
\newcommand{\N}{\mathbb{N}}
\newcommand{\Z}{\mathbb{Z}}
\newcommand{\C}{\mathbb{C}}
\newcommand{\Id}{\mathbbm{1}}

\renewcommand{\Re}{\operatorname{Re}}

\newcommand{\Var}{\mathrm{Var}}
\newcommand{\Cov}{\mathrm{Cov}}
\newcommand{\Or}{{\cal O}}

\DeclareMathOperator{\sgn}{sgn}
\DeclareMathOperator*{\Pf}{\mathrm{Pf}}

\newcommand{\D}{\mathcal{D}}

\newlength\squareheight
\newcommand\squareslash{\,\tikz{\draw (0,0) rectangle (\squareheight,\squareheight);\draw(0,0) -- (\squareheight,\squareheight)}}

\newtheorem{prop}{Proposition}[section]
\newtheorem{thm}[prop]{Theorem}
\newtheorem{lem}[prop]{Lemma}

\newtheorem{cor}[prop]{Corollary}

\newtheorem{cla}[prop]{Claim}

\newtheorem{rem}[prop]{Remark}
\newenvironment{remark}{\begin{rem}\normalfont}{\end{rem}}

\title{Time-time covariance for last passage\\ percolation in half-space}
\author{Patrik L.\ Ferrari\thanks{Institute for Applied Mathematics, Bonn University, Endenicher Allee 60, 53115 Bonn, Germany. E-mail: {\tt ferrari@uni-bonn.de}}
\and Alessandra Occelli\thanks{ENS de Lyon UMPA, 46 all\'ee d'Italie, 69007 Lyon, France. E-mail: {\tt alessandra.occelli@ens-lyon.fr}}}

\date{April 14, 2022}

\begin{document}
\sloppy
\maketitle
\begin{abstract}
This article studies several properties of the half-space last passage percolation, in particular the two-time covariance. We show that, when the two end-points are at small macroscopic distance, then the first order correction to the covariance for the point-to-point model is the same as the one of the stationary model. In order to obtain the result, we first derive comparison inequalities of the last passage increments for different models. This is used to prove tightness of the point-to-point process as well as localization of the geodesics. Unlike for the full-space case, for half-space we have to overcome the difficulty that the point-to-point model in half-space with generic start and end points is not known.
\end{abstract}

\section{Introduction}\label{sectIntro}
In this paper we consider a model in the Kardar-Parisi-Zhang (KPZ) universality class~\cite{KPZ86}, namely the last passage percolation in the \emph{half-space geometry}. In the half-space geometry results are more limited than in the full-space analogue, but at the same time the system is richer since there is one free parameter tuning the behaviour at the boundary.

The one-point limiting distribution has been first studied in a special case using Pfaffian point processes in~\cite{BR99b,BR01b,Bai02}: there one sees a transition distribution from $\rm GSE$- to $\rm GOE$-Tracy-Widom to Gaussian modulated by the free parameter when the end-point is on the diagonal. Using Pfaffian techniques, for a point-to-point model the limiting process has been characterized~\cite{SI03,BBCS17,BBCS17b,BBNV18} as well as the one for a stationary model~\cite{BFO21}. Furthermore, due to the enormous progresses in integrable probability, for a larger class of models in half-space the one-point distribution has been analyzed~\cite{BBC18,KLD19,NKDT20,BKLD20,BL21} with other properties as well~\cite{Kim21,KL18,GL12,IT18,Par19,Wu18}.

At the same time there has been an intense activity in the study of the time-time processes, mainly in the full-space setting, both in theoretical physics and in mathematics, see~\cite{Jo15,Jo18,JR19,BL16,BL17,Liu19,ND17,ND18,NDT17,FS16,FO18,JR20,BG18,CGH19,BGZ19}. In particular, in~\cite{FO18} we have shown that when the two macroscopic time are close to each other, then the first order correction of the time-time covariance is given by the variance of the Baik--Rains distribution. This was a confirmation of a prediction by Takeuchi~\cite{Tak13} and Ferrari--Spohn~\cite{FS16}.

Motivated by the progresses in the area of half-space KPZ models and of the time-time process in full-space, we study in this paper the time-time covariance of a stationary LPP model in half-space as well as the first order correction of the covariance for the point-to-point model analyzed in~\cite{BBCS17}. In this respect, the main results are: an exact formula for the covariance of the stationary model (Theorem~\ref{thmCovStatFormula}) and a formula for the covariance of the point-to-point LPP, which shows that the first correction of the covariance is the same as for the stationary model (Theorem~\ref{thmCovPPGeneral}). One of the main difference with the full-space case, is that the general point-to-point half-space LPP, with both initial and final points away from the diagonal, has not been solved yet. This implied that we could not first take the scaling limit and then analyze the behaviour for small macroscopic time difference. Thus the argument needs to be modified and all the estimates we have are uniform in the system size.

To get the results for the time-time covariance, we first develop in Section~\ref{SectComparisonInequalities} several comparison lemmas, which allow to control the increments of the point-to-point half-space LPP with stationary ones, see Proposition~\ref{PropComparison}, Corollary~\ref{CorCrossingA} and Proposition~\ref{PropCrossing}. A crucial difference with respect to the full-space case is that once two geodesics cross in the ``bulk'' of the system, they can still touch the diagonal. To obtain the results, we need some inputs from the behaviours of the geodesics, in particular on the probability that they touch the diagonal, which in turn requires estimates on upper and lower tails of related LPP models. Upper tail estimates are obtained using the Fredholm Pfaffian expansion, while for lower tail estimates we need to use Riemann-Hilbert methods (see Appendix~\ref{appRHP}).

Once we have the comparison lemmas, we can upgrade the convergence of the limit process for point-to-point LPP of~\cite{BBCS17} from finite-dimensional to weak convergence (Theorem~\ref{thmWeakCvg}). Finally, we prove the localization of geodesics in a $\Or(N^{2/3})$ neighborhood of the diagonal (Theorem~\ref{thmLocalization}).

\paragraph{Acknoledgments}  The work of P.L. Ferrari was partly funded by the Deutsche Forschungsgemeinschaft (DFG, German Research Foundation) under Germany’s Excellence Strategy - GZ 2047/1, projekt-id 390685813 and by the Deutsche Forschungsgemeinschaft (DFG, German Research Foundation) - Projektnummer 211504053 - SFB 1060. The work of A. Occelli was supported in part by ERC-2019-ADG Project 884584 LDRam, and was partially developed while A.O. was a postdoctoral fellow at MSRI during the Program “Universality and Integrability in Random Matrix Theory and Interacting Particle Systems”.\\
We are grateful to J.~Baik for the detailed explanation on how to set-up the correct Riemann-Hilbert problem and to M.~Duits, T.~Krieckerbauer and T.~Bochner for various discussions on the Riemann-Hilbert techniques. We are also grateful to G.~Barraquand for exchanges concerning their work on half-space LPP.

\section{Model and main results}\label{SectModel}
\paragraph{The models.} The last passage percolation (LPP) model with randomness $\omega$ from a point $a$ to a point $b$ is given by
\begin{equation}\label{eq1.1}
L(a,b)=\max_{\pi:a\to b} \sum_{(i,j)\in \pi\setminus \{a\}} \omega_{i,j},
\end{equation}
where the paths are up-right paths. Paths which maximize \eqref{eq1.1} are called \emph{geodesics}. When $a=(0,0)$ we will just write $L(b)$ for $L(a,b)$. Removing the initial point is just for convenience, so that we have the concatenation property
\begin{equation}
L(a,b)=\max_{c\in I}\{L(a,c)+L(c,b)\}
\end{equation}
where $I$ is any downright path separating $a$ and $b$. In particular, we always have the inequality
\begin{equation}\label{eqSubadd}
L(a,b)\geq L(a,c)+L(c,b),
\end{equation}
with equality iff $c$ belongs to a geodesic of the LPP.

In this paper we consider LPP on half-space, that is, where $\omega_{i,j}$ can be non-zero only for $i\geq j\geq 0$. For that define the following regions
\begin{equation}\label{new_boundary}
\left\{
\begin{aligned}
{\cal R} &= \{(m,n)\in\Z^2\mid 0\leq n = m \text{ or } 0=n\leq m\},\\
{\cal B} &=\{(m,n)\in \Z^2\mid 0<n<m\},\\
\D &=\{(m,n)\in \Z^2\mid 1\leq n=m\}.
\end{aligned}
\right.
\end{equation}
The region $\cal B$ is what we think as bulk of the system and will have the same randomness for all different half-space models, which differ only by the randomness on the boundary regions $\cal R$, of which the diagonal $\cal D$ is a subset.

The two model we will deal are a point-to-point and a stationary LPP given as follows\footnote{This is not the only stationary LPP in half-space, but the one with constant increments, which corresponds in the exclusion process analogue to an input rate of particle chosen such that the average density of particle is constant. In this case, the measure is product measure. In general one would expect at least a two-parameter family. Other stationary measures are known~\cite{Gro04} for TASEP and very recently~\cite{BC22} for LPP.}. Let us consider the following LPP models.
For $\rho\in(0,1)$, the \emph{stationary} model with parameter $\rho$, denoted by $L^{\rm st,\rho}$, has weights~\cite{BFO20}
\begin{equation}\label{eqWeightsStat}
\left\{
\begin{array}{ll}
\omega^\rho_{0,0} = 0, &\\
\omega^\rho_{i,i} \sim \exp(\rho), &\text{ for }i\in\N,\\
\omega^\rho_{i,0} \sim \exp(1-\rho), &\text{ for }i\in\N, \\
\omega^\rho_{i,j} \sim \exp(1), &\text{ for }(i,j)\in {\cal B}.
\end{array}
\right.
\end{equation}
The \emph{point-to-point} model, denoted by $L^{\rm pp}$, has weights
\begin{equation}\label{eqWeightsPP}
\left\{
\begin{array}{ll}
\omega^{pp}_{0,0} = 0, &\\
\omega^{pp}_{i,i} \sim \exp(1/2+\alpha), &\text{ for }i\in\N,\\
\omega^{pp}_{i,0} = 0, &\text{ for }i\in\N, \\
\omega^{pp}_{i,j} \sim \exp(1), &\text{ for }(i,j)\in {\cal B}.
\end{array}
\right.
\end{equation}
The dependence on the parameter $\alpha$ in $L^{\rm pp}$ is implicit. See Figure~\ref{fig:ptp_stat}.
\begin{figure}[t!]
 \centering
\includegraphics[height=5cm]{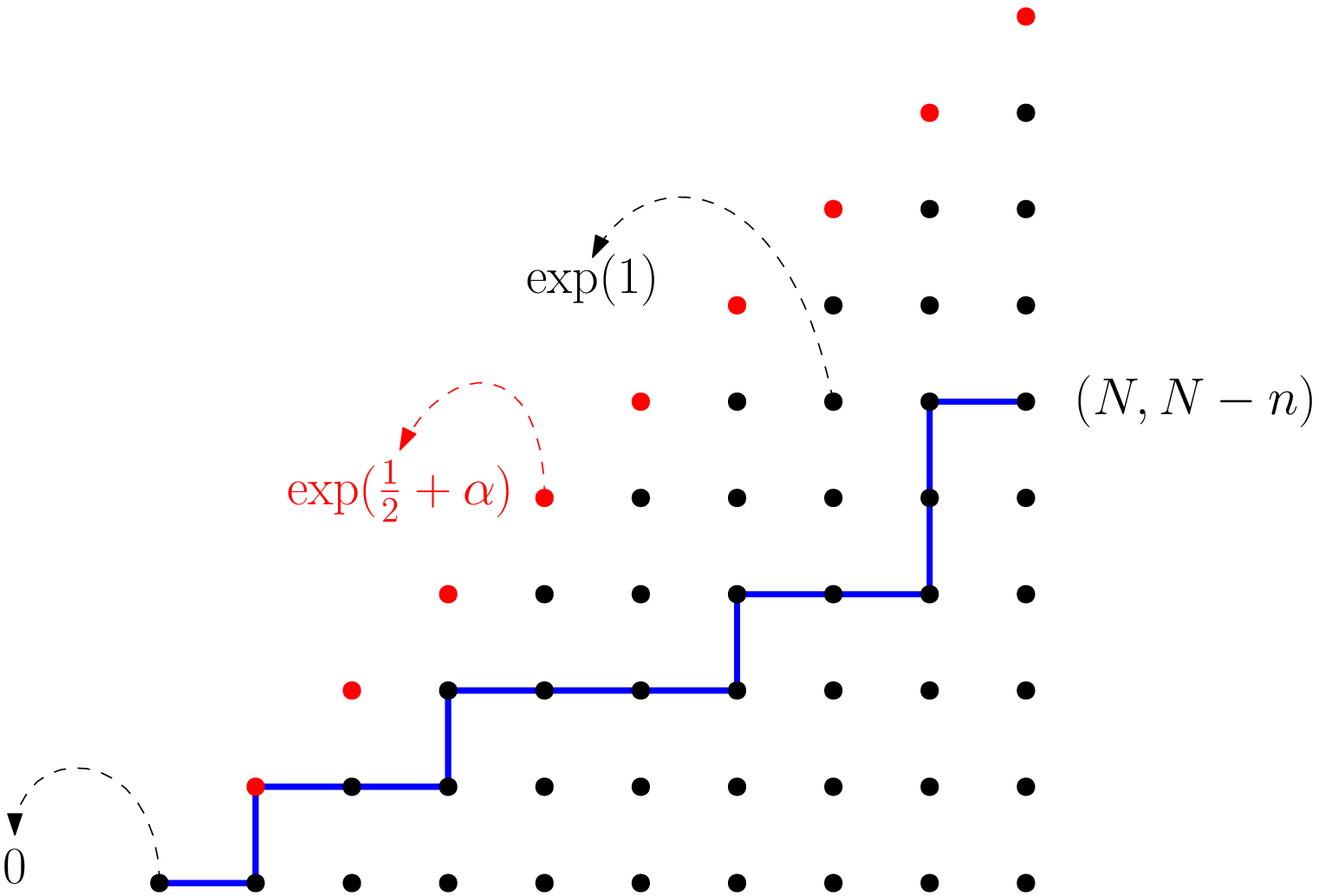}%
 \includegraphics[height=5cm]{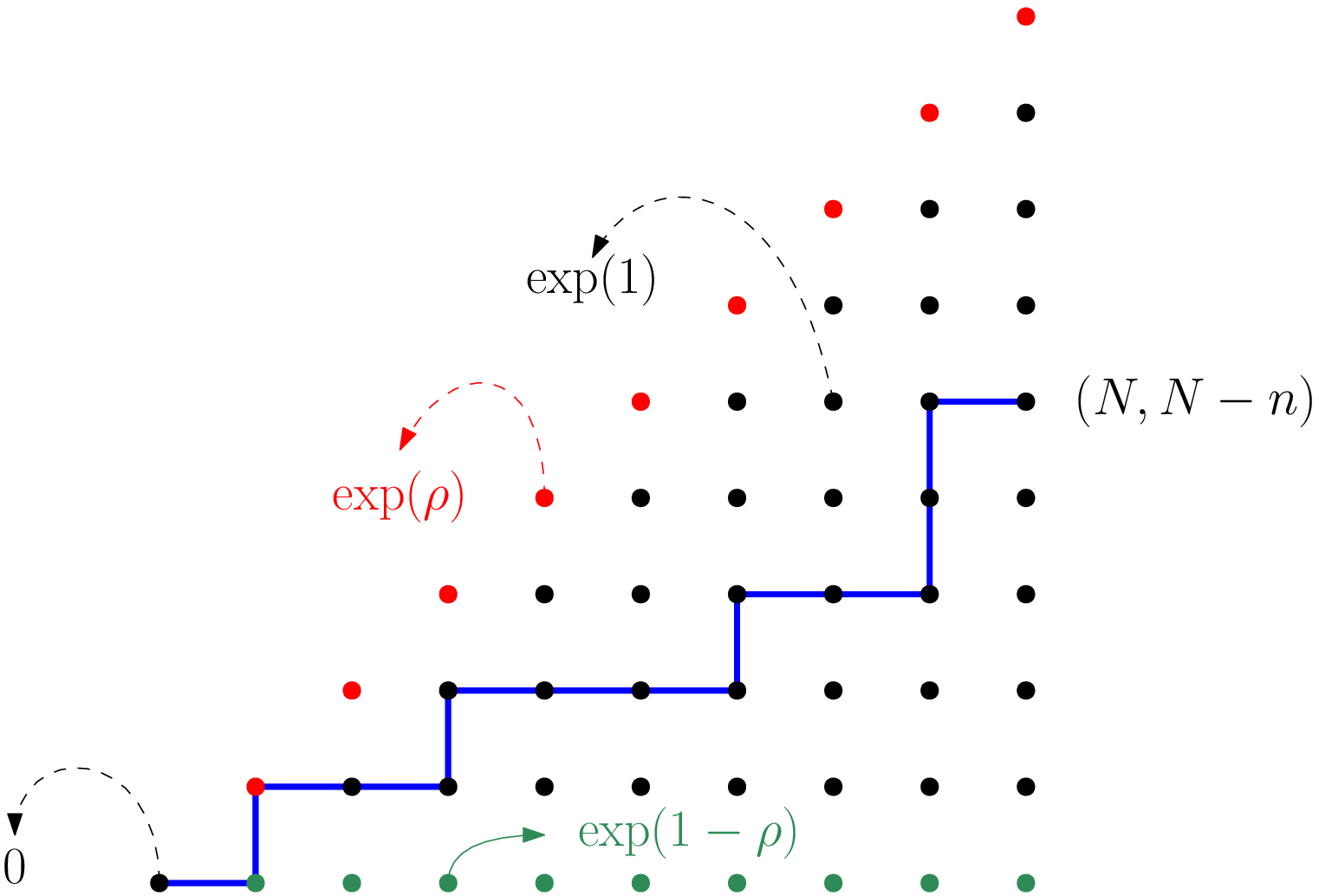}
\caption{Left: a point-to-point LPP with weigths distributed as in~\eqref{eqWeightsPP}.\\ Right: a stationary LPP with boundary parameter $\rho$ as in~\eqref{eqWeightsStat}.}
\label{fig:ptp_stat}
\end{figure}

We consider any coupling between these models having the same weights in the bulk ${\cal B}$,
\begin{equation}
    \omega_{i,j} = \omega^{pp}_{i,j} = \omega^\rho_{i,j}, \quad\text{ for }(i,j)\in {\cal B}\cup{(0,0)} \text{ and }\rho\in(0,1),
\end{equation}
and with the monotonicity condition on the weights in the diagonal
\begin{equation}
\left\{
\begin{aligned}
\omega^{pp}_{i,i} &\leq \omega^\rho_{i,i}, &\text{ for }i\in\N \text{ and }\rho\in(0,1/2+\alpha],\\
\omega^{pp}_{i,i} &\geq \omega^\rho_{i,i}, &\text{ for }i\in\N \text{ and }\rho\in[1/2+\alpha,1).
\end{aligned}
\right.
\end{equation}
Furthermore, when we consider two stationary models with densities $\rho_-<\rho_+$, then they are coupled to have the same weights in the bulk and at the boundaries they satisfy
\begin{equation}\label{eq:coupling}
\omega_{i,i}^{\rho_-}\geq \omega_{i,i}^{\rho_+},\quad \omega_{i,0}^{\rho_-}\leq \omega_{i,0}^{\rho_+},\quad i\geq 1.
\end{equation}

\paragraph{Some known scaling limits.}
Let us first state the scaling limits and some know results. Let $\rho=\frac12+\delta 2^{-4/3} N^{-1/3}$ be density for the point-to-point LPP or also for the stationary one. Consider the end-points
\begin{equation}
\begin{aligned}
Q_1&=(N+M_1 (2N)^{2/3},N-M_1(2N)^{2/3}),\quad M_1=(1-\tau)^{2/3}\tilde M_1\\
Q_\tau&=(\tau N+M_\tau (2N)^{2/3},\tau N-M_\tau(2N)^{2/3}),\quad M_\tau=(1-\tau)^{2/3} \tilde M_\tau.
\end{aligned}
\end{equation}
Here the parameter $\tau$ has the meaning of the macroscopic time variable, while changing the value of $M_1$, $M_\tau$, corresponds to looking at the process in space.
Then
\begin{equation}\label{eq5.2}
\begin{aligned}
{\cal L}_N^{\rm pp}(M_1,1)&:=\frac{L^{\rm pp}(Q_1)-4N}{2^{4/3}N^{1/3}}\stackrel{N\to\infty}{\longrightarrow}{\cal A}^{\rm pp}_\delta(M_1)-M_1^2,\\
{\cal L}_N^{\rm pp}(M_\tau,\tau)&:=\frac{L^{\rm pp}(Q_\tau)-4\tau N}{2^{4/3}N^{1/3}}\stackrel{N\to\infty}{\longrightarrow} \tau^{1/3} \tilde {\cal A}^{\rm pp}_{\delta\tau^{1/3}}(M_\tau/\tau^{2/3})-M_\tau^2/\tau,
\end{aligned}
\end{equation}
where ${\cal A}^{\rm pp}_{\delta}$ and $\tilde {\cal A}^{\rm pp}_{\delta\tau^{1/3}}$ are the limit point-to-point process derived in~\cite{BBCS17}, see \eqref{eqPPprocess}. In Theorem~\ref{thmWeakCvg} we lift the convergence from finite-dimensional to weak convergence on the space of continuous functions on compact intervals.

Similarly, with $Q(u)=(N+u (2N)^{2/3},N-u (2N)^{2/3})$,
\begin{equation}\label{eq2.12}
\frac{L^{\rm st,\rho}(Q(u))-4N}{2^{4/3}N^{1/3}}\stackrel{N\to\infty}{\longrightarrow} {\cal A}^{\rm st,hs}_\delta(u),
\end{equation}
where ${\cal A}^{\rm st,hs}_\delta(u)$ is the half-space stationary limit process described in~\cite{BFO21}. This process is normalized to have
\begin{equation}\label{eqMeanStatProc}
 \E({\cal A}^{\rm st,hs}_\delta(u))=\delta(2u+\delta).
\end{equation}
Rescaling \eqref{eq2.12} we obtain
\begin{equation}\label{eq5.4}
\begin{aligned}
{\cal L}_N^{\rm st,\rho}(M_1,1)&:=\frac{L^{\rm st,\rho}(Q_1)-4N}{2^{4/3}N^{1/3}}\stackrel{N\to\infty}{\longrightarrow} {\cal A}^{\rm st,hs}_\delta(M_1),\\
{\cal L}_N^{\rm st,\rho}(M_\tau,\tau)&:=\frac{L^{\rm st,\rho}(Q_\tau)-4\tau N}{2^{4/3}N^{1/3}}\stackrel{N\to\infty}{\longrightarrow} \tau^{1/3}{\cal A}^{\rm st,hs}_{\delta\tau^{1/3}}(M_\tau).
\end{aligned}
\end{equation}
Finally we have the identity
\begin{equation}\label{eq5.5}
\begin{aligned}
{\rm Cov}({\cal L}_N^{*}(M_1,1),{\cal L}_N^{*}(M_\tau,\tau))=& \tfrac12 \Var({\cal L}_N^{*}(M_1,1))+\tfrac12\Var({\cal L}_N^{*}(M_\tau,\tau))\\
&-\tfrac12\Var({\cal L}_N^{*}(M_\tau,\tau)-{\cal L}_N^{*}(M_1,1)).
\end{aligned}
\end{equation}
where $*\in\{{\rm pp};{\rm st},\rho\}$. The first two variances in the r.h.s.~of \eqref{eq5.5} converges to the corresponding limits due to the tail estimates of Appendix~\ref{AppRoughBounds}. Thus the interesting term we have to analyze is the last one.

\paragraph{Main results on the two-time covariance.} The following results are proven in Section~\ref{SectCovariance}.
The first result is an exact formula for the stationary LPP with end-points on the diagonal, i.e., $M_1=M_\tau=0$.
\begin{thm}\label{thmCovStatFormula}
Let $\rho=\frac12+\delta 2^{-4/3} N^{-1/3}$ and ${\cal L}_N^{\rm st,\rho}$ as in \eqref{eq5.4}. Then
\begin{equation}
\begin{aligned}
\lim_{N\to\infty}{\rm Cov}({\cal L}_N^{\rm st,\rho}(0,1),{\cal L}_N^{\rm st,\rho}(0,\tau))=&\tfrac12 \Var({\cal A}^{\rm st,hs}_\delta(0))+\tfrac12 \tau^{2/3}\Var({\cal A}^{\rm st,hs}_{\delta\tau^{1/3}}(0))\\
&-\tfrac12 (1-\tau)^{2/3}\Var({\cal A}^{\rm st,hs}_{\delta(1-\tau)^{1/3}}(0)).
\end{aligned}
\end{equation}
\end{thm}
To get this result, the following variational identity is derived,
 \begin{equation}\label{eq5.12}
 \max_{v\geq 0} \left\{\sqrt{2}B(v)+2v\delta+{\cal A}^{\rm pp}_{\delta}(v)-v^2\right\} \stackrel{(d)}{=} {\cal A}^{\rm st,hs}_\delta(0).
\end{equation}

\begin{rem}
For points away from the diagonal, we do not have a closed formula in terms of known random variables, because the generic point-to-point half-space LPP is not yet known.
\end{rem}

The second result is about the universal asymptotic behaviour of the last term in \eqref{eq5.5} when $\tau\to 1$.
\begin{thm}\label{thmCovPPGeneral}
Let $\rho=\frac12+\delta 2^{-4/3} N^{-1/3}$, ${\cal L}_N^{\rm pp}$ as in \eqref{eq5.2} and ${\cal L}_N^{\rm st,\rho}$ as in \eqref{eq5.4}. Then, for any $0<\eta<1/15$, there exists a constant $C$ such that
\begin{equation}
\lim_{N\to\infty} |{\rm Var}({\cal L}_N^{\rm pp}(M_1,1)-{\cal L}_N^{\rm pp}(M_\tau,\tau))-{\rm Var}({\cal L}_N^{\rm st,\rho}(M_1,1)-{\cal L}_N^{\rm st,\rho}(M_\tau,\tau))| \leq C (1-\tau)^{11/15-\eta}
\end{equation}
 as $\tau\to 1$.
\end{thm}

For special case $M_\tau=0$, we get a better error term estimate, see Theorem~\ref{thmCovPP}, namely $\Or((1-\tau)^{1-\theta})$ for $0<\theta<1/3$. To get a result of the same precision as Theorem~\ref{thmCovPP}, we would need to get an optimal bound on
\begin{equation}
\Theta=\max_{0\leq u\leq M} \left|({\cal L}_N^{\rm st,\rho}(u,\tau)-{\cal L}_N^{\rm st,\rho}(M_\tau,\tau))-({\cal L}_N^{\rm st,\rho_-}(u,\tau)-{\cal L}_N^{\rm st,\rho_-}(M_\tau,\tau))\right|.
\end{equation}
This requires to know something about the coupling between different stationary models in half-space. Results in this directions are not yet available (unlike for the full-space case~\cite{FS20}).

As a corollary, for the special case $M_1=M_\tau=0$, we have an explicit formula for the first order expansion in the point-to-point case as $\tau\to 1$, compare with the recent paper on half-space KPZ equation, Section~1.4 of~\cite{BKLD22} as well.
\begin{cor}
Let $\rho=\frac12+\delta 2^{-4/3} N^{-1/3}$ and ${\cal L}_N^{\rm pp}$ as in \eqref{eq5.2}. Then, as $\tau\to 1$, for $0<\theta<1/3$,
\begin{equation}
\begin{aligned}
\lim_{N\to\infty}{\rm Cov}({\cal L}_N^{\rm pp}(0,1),{\cal L}_N^{\rm pp}(0,\tau))=&\tfrac12 \Var({\cal A}^{\rm pp}_\delta(0))+\tfrac12 \tau^{2/3}\Var({\cal A}^{\rm pp}_{\delta\tau^{1/3}}(0))\\
&-\frac12 (1-\tau)^{2/3}\Var({\cal A}^{\rm st,hs}_{\delta(1-\tau)^{1/3}}(0))+\Or((1-\tau)^{1-\theta}).
\end{aligned}
\end{equation}
\end{cor}

\begin{remark}
If we would consider $M_1>0$ not scaled in $\tau$, and $M_\tau=M_1+\tilde M_\tau(1-\tau)^{2/3}$, then as $\tau\to 1$, the geodesic from time $\tau N$ to time $N$ will not touch the diagonal anymore, so that the correction term will be given by $-\frac12(1-\tau)^{2/3} \Var(\xi_{\rm BR})$ where $\xi_{\rm BR}$ is a Baik-Rains distribution function (with parameter depending on $\tilde M_\tau$), like for the full-space case.
\end{remark}

\section{Comparison inequalities for half-space LPP}\label{SectComparisonInequalities}
In this section we obtain comparison inequalities for the half-space LPP, see Propositions~\ref{PropComparison} and~\ref{PropComparisonAlternative}. We then apply them to be able to compare the increments of the point-to-point LPP with stationary models, see Corollary~\ref{CorCrossingA} and Proposition~\ref{PropCrossing}.

\subsection{Comparison results for half-space LPP}

The first comparison result, Proposition~\ref{PropComparison}, is about the increments of LPP which differs only on the randomness on $\cal R$. Unlike in the full-space case, a geodesic can visit both $\cal D$ and ${\cal R}\setminus {\cal D}$. This implies some modifications with respect to the analogue result in full-space~\cite{CP15b,Pim17,FGN17}.

For two points $p,q\in\Z^2$ we denote
\begin{equation}
 p\preceq q \quad \Leftrightarrow \quad p_1\leq q_1\textrm{ and }p_2\geq q_2.
\end{equation}
Furthermore, for two paths $\pi,\tilde\pi$ in $\Z^2$ we write $\pi\preceq \tilde \pi$ if for any down-right path $\cal Y$, ${\cal Y}\cap \pi\preceq {\cal Y}\cap \tilde\pi$ (whenever the intersections are non-empty).

Consider LPP with different boundary conditions with randomness $\widetilde\omega$ and $\omega$ coupled by setting the same randomness in the bulk, that is, by the condition $\widetilde\omega_{i,j}=\omega_{i,j}$ for all $(i,j)\in{\cal B}$. Denote by $\widetilde L$ and $L$ the respective LPP, and the geodesic to a point $p$ will be denoted by $\widetilde \pi(p)$ and $\pi(p)$ respectively.

\begin{prop}\label{PropComparison} Consider two end-points $p,q \in {\cal B}$ such that $p\preceq q$.
Assume that $\widetilde \pi(p)\cap \pi(q)\cap {\cal B}\neq\emptyset$. If at least one of the following conditions are satisfied
\begin{itemize}
\item[(a)] $\widetilde{\omega}_{i,i} \leq \omega_{i,i}$ for all $i\in\N$,
\item[(b)] $\widetilde{\omega}_{i,i} \geq \omega_{i,i}$ for all $i\in\N$ and $\widetilde \pi(p)\cap {\cal D}=\emptyset$,
\end{itemize}
then
\begin{equation}
    L(q)-L(p) \leq \widetilde{L}(q)-\widetilde{L}(p).
\end{equation}
\end{prop}
\begin{proof} Denote the increments of the LPP from $p$ to $q$ by
\begin{equation}
\Delta L= L(q)-L(p),\quad
\Delta \widetilde{L} = \widetilde{L}(q)-\widetilde{L}(p).
\end{equation}
Let $c \in {\cal B}$ be the last crossing point, that is the points in $\widetilde \pi(p)\cap \pi(q)$ which is farther from the origin (in $L^\infty$ distance).
Since $c$ belongs to the geodesics of $\widetilde L(p)$ and of $L(q)$, we have
\begin{equation}\label{equa2}
\widetilde{L}(p)= \widetilde{L}(c) + \widetilde{L}(c,p),\quad
L(q) = L(c) + L(c,q).
\end{equation}

On the other hand, by \eqref{eqSubadd}, we have the inequalities
\begin{equation}\label{ineq2}
\widetilde{L}(q) \geq \widetilde{L}(c) + \widetilde{L}(c,q),\quad
L(p) \geq L(c) + L(c,p).
\end{equation}
By combining \eqref{equa2} and \eqref{ineq2}, we get
\begin{equation}
\Delta L \leq L(c, q) - L(c, p),\quad \Delta \widetilde{L} \geq \widetilde{L}(c, q) - \widetilde{L}(c, p).
\end{equation}
Unlike in the full-space LPP, for half-space LPP the bounds on the increments do not match exactly because the paths after $c$ might still touch the boundary at the diagonal $\cal D$.

If condition (a) is satisfied, then by the monotonicity condition on the diagonal we have the inequalities
\begin{equation}\label{genauso}
\widetilde{L}(c, p) \leq L(c, p),\quad
\widetilde{L}(c, q) \leq L(c, q).
\end{equation}
Let us show that the second inequality in \eqref{genauso} is in fact an equality. First of all, note that since $\omega_{i,i}\geq \widetilde\omega_{i,i}$, we have $\pi(c,q)\preceq \tilde \pi(c,q)$. Then, $\widetilde{L}(c, q)< L(c, q)$ if and only if $\pi(c,q)\cap {\cal D}\neq\emptyset$. Let us see that this can not happen. Assume that $d=\pi(c,q)\cap {\cal D}$ exists. However, the next point of the geodesics $\pi(c,q)$ and $\tilde \pi(c,p)$ are both given by $d+(1,0)$ lies into ${\cal B}$, which is a contradiction of the assumption that $c$ is the last intersection point. Therefore we have shown that
\begin{equation}
\widetilde{L}(c,q)=L(c,q).
\end{equation}
Putting all together we obtain
\begin{equation}
\Delta L\leq L(c, q)-L(c, p) \leq \widetilde L(c, q)-\widetilde L(c, p)\leq \Delta\widetilde L.
\end{equation}

If condition (b) is satisfied, then by monotonicity of the weights on the diagonal we have
\begin{equation}
\tilde \pi(c,q)\preceq \pi(c,q),\quad \tilde \pi(c,p)\preceq \pi(c,p)
\end{equation}
and by order of geodesics we have
\begin{equation}
\tilde \pi(d,p)\preceq \tilde \pi(d,q),\quad \pi(d,p)\preceq \pi(d,q),
\end{equation}
and $\tilde \pi(p)\preceq \tilde \pi(q)$. This last ordering implies that $\tilde \pi(p)\preceq \tilde \pi(c,p)$ as well. So, under the conditions on the diagonal weights in (b) we have
\begin{equation}
\tilde \pi(p)\preceq\{\tilde \pi(c,p),\tilde \pi(c,q),\pi(c,p),\pi(c,q)\}.
\end{equation}
In addition, if $\tilde \pi(p)\cap {\cal D}=\emptyset$, then this implies that
\begin{equation}
\widetilde{L}(c, p) = L(c, p),\quad
\widetilde{L}(c, q) = L(c, q),
\end{equation}
which gives $\Delta L \leq \Delta\widetilde L$.
\end{proof}

\begin{remark}
Unlike for the full-space geometry, here we need to put extra conditions to satisfy the inequalities. The reason is that the geodesics from the intersection point to the end points can still touch the diagonal and thus the associated LPP for the two conditions are different. Condition (b) can be useful only when the end-point $p$ is far enough from the diagonal so that effectively the weights on the diagonal are not used. In the rest of the paper we did not apply Proposition~\ref{PropComparison} (b), but we keep it since it could potentially be of use in other works.
\end{remark}

\begin{figure}[t!]
  \centering
   \includegraphics[height=5cm]{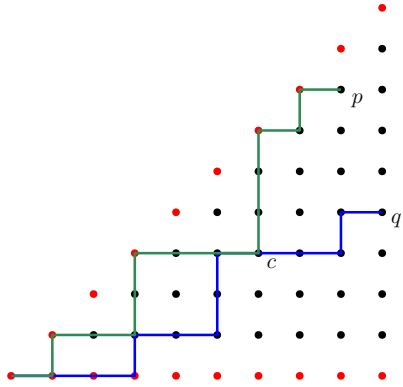}
\caption{Geometric setting of the comparison result of Proposition~\ref{PropComparison}. The green path is $\tilde \pi(p)$ and the blue one is $\pi(q)$. The boundary $\cal R$ is indicated by the red dots, while the bulk $\cal B$ by the black dots.}
\label{fig_PropComparison}
\end{figure}

A second comparison result is obtained when the randomness on the diagonal as well on the horizontal axis are coupled in such a way that there is a certain ordering, see Lemma~B.1 of~\cite{BBS21} for the analogue result in the full-space geometry.
\begin{prop} \label{PropComparisonAlternative}
Consider two end-points $p,q \in {\cal B}$ such that $p\preceq q$. Assume that the randomness are coupled on the boundaries such that
\begin{equation}
 \widetilde{\omega}_{i,i} \leq \omega_{i,i}, \quad \widetilde{\omega}_{i,0} \geq \omega_{i,0} \quad \textrm{ for all } i\geq 1.
\end{equation}
Then we have
\begin{equation}\label{eq:PropCompAlt}
    L(q)-L(p) \leq \widetilde{L}(q)-\widetilde{L}(p).
\end{equation}
\end{prop}
\begin{proof}
With the choice of the weights, the LPP $\widetilde L$ has smaller weights on the diagonal and larger weights on the horizontal axis with respect to the LPP $L$. As a consequence we have the order of geodesics: for any point $p$, $\pi(p)\preceq\tilde\pi(p)$.

We prove the statement by contradiction. Let us assume that~\eqref{eq:PropCompAlt} is not true. Then, there will be a point $r$ such that for points to its left or below it, \eqref{eq:PropCompAlt} is satisfied, but either
 \begin{equation*}
 (a)\qquad  L(r+e_1)-L(r) > \widetilde{L}(r+e_1)-\widetilde{L}(r)
 \end{equation*}
 and/or
 \begin{equation*}
(b)\qquad L(r)-L(r+e_2) > \widetilde{L}(r+e_2)-\widetilde{L}(r)
 \end{equation*}
hold.
Let us consider the situation when $(a)$ holds. The other case is completely analogous and we omit the proof.

Consider the last step of the geodesics ending at $r+e_1$. Due to the order of geodesics, the only possible cases are:
 \begin{itemize}
  \item[$(i)$] Both $\pi(r+e_1)$ and $\tilde\pi(r+e_1)$ cross the point $r$ before reaching $r+e_1$. In this case, we have $L(r+e_1)-L(r) =\omega_{r+e_1}= \widetilde{L}(r+e_1)-\widetilde{L}(r)$, which contradicts $(a)$.
  \item[$(ii)$] $\pi(r+e_1)$ crosses $r$ and $\tilde\pi(r+e_1)$ crosses $r+e_1-e_2$. Then, we have
   \begin{equation}
   L(r+e_1)=L(r)+\omega_{r+e_1},\quad \widetilde L(r+e_1)\geq \widetilde L(r)+\omega_{r+e_1},
   \end{equation}
   which imply
   \begin{equation}
   L(r+e_1)-L(r)=\omega_{r+e_1}\leq \widetilde L(r+e_1)- \widetilde L(r),
   \end{equation}
   contradicting $(a)$.
  \item[$(iii)$] Both $\pi(r+e_1)$ and $\tilde\pi(r+e_1)$ cross the point $r+e_1-e_2$ before reaching $r+e_1$. This implies
 \begin{equation}
 L(r+e_1-e_2)-L(r-e_2)= \widetilde L(r+e_1-e_2)-\widetilde L(r-e_2).
\end{equation}
By the definition of $r$, \eqref{eq:PropCompAlt} holds for $q=r-e_2$ and $p=r$, namely
\begin{equation}
   L(r-e_2)-L(r)\leq \widetilde L(r-e_2)-\widetilde L(r).
\end{equation}
Finally, recall that assumption (a) gives
\begin{equation}
 L(r+e_1)-L(r)> \widetilde L(r+e_1)-\widetilde L(r).
\end{equation}
These last three equations lead to
\begin{equation}
 \begin{aligned}
   L(r+e_1)-L(r-e_2)=&  L(r+e_1)-L(r)+ L(r)-L(r-e_2)\\
   >& \widetilde L(r+e_1)-\widetilde L(r)+\widetilde L(r)-\widetilde L(r-e_2)\\
   =& \widetilde L(r+e_1)-\widetilde L(r-e_2),
 \end{aligned}
\end{equation}
and
\begin{equation}
 \begin{aligned}
   L(r+e_1)-L(r-e_2)=&  L(r+e_1)-L(r+e_1-e_2)+ L(r+e_1-e_2)-L(r-e_2)\\
   \leq& \widetilde L(r+e_1)-\widetilde L(r+e_1-e_2)+\widetilde L(r+e_1-e_2)-\widetilde L(r-e_2)\\
   =& \widetilde L(r+e_1)-\widetilde L(r-e_2).
 \end{aligned}
\end{equation}
This leads to a contradiction.
 \end{itemize}
\end{proof}

\subsection{Bounds on probabilities of geodesic crossings}
With the above mentioned coupling between LPP, Proposition~\ref{PropComparisonAlternative} gives a simple bound on the upper bound of the point-to-point LPP.
\begin{cor}\label{CorCrossingA} Let $\rho_+\geq \rho$. Consider the stationary LPP with parameter $\rho_+$ and the point-to-point model with parameter $\rho$. Then for all $p\preceq q$, th
\begin{equation}
    L^{\rm pp}(q)-L^{\rm pp}(p) \leq L^{\rho_+}(q)-L^{\rho_+}(p).
\end{equation}
\end{cor}
Furthermore, for two coupled stationary initial condition, we have monotonicity in the increments.
\begin{cor}\label{CorComparisonStatStat}
Let $\rho_-<\rho_+$ be two parameters of stationary models and the LPP coupled as above. Then for all $p\preceq q$,
\begin{equation}
     L^{\rho_-}(q)-L^{\rho_-}(p) \leq L^{\rho_+}(q)-L^{\rho_+}(p).
\end{equation}
\end{cor}
For the full-space LPP there is a special coupling between stationary models with different densities, such that the coupling is the same for each line~\cite{FS20}. An analogue result would be welcome in the half-space, since it would allow to improve the error term to the first order of the covariance studied in Section~\ref{SectCovariance}.

Finally, let us mention one small inequality between half-space LPP and full-space LPP. Let us denote by $L^{\square}$ the LPP with $\omega_{i,j}\sim\exp(1)$, $i,j\geq 1$. Let $L^{\rm pp;1}$ be the half-space LPP with parameter $\rho=1$ on the diagonal. Couple $L^{\square}$ and $L^{\rm pp;1}$ by assuming that the randomness in for $j\geq i\geq 1$ are identical.
\begin{lem}\label{lemCompFullSpaceHalfSpace}
For all $p\preceq q$ with ${\cal D}\preceq p$,
\begin{equation}
L^{\square}(q)-L^{\square}(p)\leq L^{\rm pp;1}(q)-L^{\rm pp;1}(p)
\end{equation}
\end{lem}
\begin{proof}
The proof is similar to the one of Proposition~\ref{PropComparison}(a). The difference is that what was called $\cal B$ is now $\widehat{\cal B}:=\{(m,n)\in\Z^2\,|\, 1\leq n\leq m\}$ and the requirements on the weights on the diagonal becomes a requirements on the weights on $\widehat {\cal D}:=\{(m,n)\in\Z^2\,|\,1\leq m < n\}$, namely
\begin{equation}
0=\omega^{\rm pp;1}_{i,j}\leq \omega^{\square}_{i,j}\textrm{ for all }(i,j)\in\widehat {\cal D}.
\end{equation}
Then the inequality follows because $\pi^{\rm pp;1}(p)\cap \pi^{\square}(q)\cap \widehat{\cal B}\neq \emptyset$ is always satisfied.
\end{proof}

We will apply Proposition~\ref{PropComparison} with one of the two LPP being the stationary model with a parameter smaller that $\rho$. The reason being that for the stationary case we exactly know the law of the increments. The central step is to get appropriate bounds on the probability of having a crossing in $\cal B$ of a stationary geodesic and the point-to-point geodesic. These are given in the Proposition~\ref{PropCrossing} below.
\begin{prop}\label{PropCrossing}
Let us consider $\alpha=\delta 2^{-4/3}N^{-1/3}$, $\rho=\tfrac12+\alpha$ and $\rho_-=\tfrac12+\alpha-\kappa 2^{-4/3} N^{-1/3}$ with $\kappa>0$. Let $u_1,u_2\in \R_{\geq0}$ such that $u_1< u_2$. Let us consider the following points\footnote{When writing a point $(x,y)$ we mean always its approximation on the $\Z^2$, i.e., $(\lfloor x\rfloor,\lfloor y \rfloor)$.}
\begin{equation}
\left\{
\begin{aligned}
p &= (N,N) + u_1 (2N)^{2/3}(1,-1),\\
q &= (N,N) + u_2 (2N)^{2/3}(1,-1).
\end{aligned}
\right.
\end{equation}
Define the crossing event $\Omega_{\rm cross}=\{\pi^{\rho_-}(q)\cap \pi^{pp}(p)\cap {\cal B}\neq \emptyset\}$. Assume $\kappa-\delta\geq \max\{1,6u_2\}$ with $\kappa-\delta=o(N^{1/12})$. Then, there exist constants $C,c>0$ such that
\begin{equation}
\Pb(\Omega_{\rm cross})\geq 1-C e^{-c (\kappa-\delta)^3}
\end{equation}
for all $N$ large enough. From Proposition~\ref{PropComparison} (a) it then follows that under the event $\Omega_{\rm cross}$ we have the inequality
\begin{equation}
    L^{\rho_-}(q)-L^{\rho_-}(p) \leq L^{\rm pp}(q)-L^{\rm pp}(p).
\end{equation}
\end{prop}

\begin{remark}\label{rem3.7}
In the proof we actually get an estimate on $\Pb(\pi^{\rho_-}(q)\cap\D\neq\emptyset)$. Thus the result holds for larger classes of half-space LPP models. We stated it only in this case since other cases have not been solved yet.
\end{remark}

For the proof of Proposition~\ref{PropCrossing} we use bounds on the upper and lower tails of different half-space LPP models. For the upper bound we first relate it with a case where a Fredholm Pfaffian representation is known and perform asymptotic analysis on the correlation kernel. The lower bound turned out to be more tricky, since we could not refer to known lower tail estimates present in the literature. However, we were able to related it with a point-to-point LPP with end-point on the diagonal, for which in the geometric setting a Riemann-Hilbert representation of the distribution function was available. Unfortunately the asymptotics we were looking for had not been worked out yet. The needed lower bound is worked out using the Riemann-Hilbert method in Appendix~\ref{appRHP}.

\begin{proof}[Proof of Proposition~\ref{PropCrossing}]
Clearly for all $N$ large enough, $p$ and $q$ are in ${\cal B}\cup \D$ and we have $p_1\leq q_1$ and $p_2\geq q_2$. Thus the result follows from Proposition~\ref{PropComparison} (a) once we have a bound on the probability of the crossing event. A sufficient condition for having the intersection is that $\pi^{\rho_-}(q)\cap \D\neq\emptyset$. Therefore
\begin{equation}
\Pb(\pi^{\rho_-}(q)\cap \pi^{pp}(p)\cap {\cal B}\neq \emptyset)\geq \Pb(\pi^{\rho_-}(q)\cap\D\neq\emptyset).
\end{equation}

Define by $L^{\rho_-}_{\D}(q)$ the LPP obtained by maximizing over all up-right paths with at least one point on $\D$ and $L^{\rho_-}_{\D^c}(q)$ the LPP obtained by maximizing over all up-right paths without points on $\D$. Then we have $L^{\rho_-}(q) = \max\{L^{\rho_-}_{\D}(q),L^{\rho_-}_{\D^c}(q)\}$. The geodesic $\pi^{\rho_-}(q)$ touches the diagonal if and only if $L^{\rho_-}(q) > L^{\rho_-}_{\D^c}(q)$. Thus we have, for any choice of $S\in\R$,
\begin{equation}\label{eq:nonintersection}
\begin{aligned}
    \Pb(\pi^{\rho_-}(q)\cap\D\neq\emptyset) &= \Pb(L^{\rho_-}(q) > L^{\rho_-}_{\D^c}(q)) \\
        &\geq \Pb(L^{\rho_-}(q)\geq S > L^{\rho_-}_{\D^c}(q)) \\
        &\geq 1 - \Pb(L^{\rho_-}(q) < S) - \Pb(L^{\rho_-}_{\D^c}(q)\geq S).
\end{aligned}
\end{equation}
We need to choose $S=S(N)$ such that the last two probabilities in the r.h.s.\ are small. By stationarity we can compute exactly the expected value of $L^{\rho_-}(q)$, which is given by
\begin{equation}\label{eq2.41}
\E(L^{\rho_-}(q))=\frac{N+u_2 (2N)^{2/3}}{1-\rho_-}+\frac{N-u_2 (2N)^{2/3}}{\rho_-}=4N+2^{4/3}N^{1/3}[(\kappa-\delta)^2-2u_2 (\kappa-\delta)]+\Or(1).
\end{equation}
We consider
\begin{equation}\label{eqChoiceOfS}
S=4N+2^{4/3}N^{1/3}[\tfrac12 (\kappa-\delta)^2-2u_2 (\kappa-\delta)].
\end{equation}
Applying Propositions~\ref{PropLowerBoundStat} and~\ref{PropUpperBoundNotDiagonal} below, we obtain the claimed result.
\end{proof}

\begin{prop}\label{PropLowerBoundStat}
Let $S$ be chosen as in \eqref{eqChoiceOfS}, and $\alpha$, $\rho_-$ as in Proposition~\ref{PropCrossing}. For any $\kappa>\delta$ with $\kappa-\delta=o(N^{1/12})$
\begin{equation}
\Pb(L^{\rho_-}(q) < S) \leq C e^{-c (\kappa-\delta)^3}
\end{equation}
for all $N$ large enough, where the constants $C,c$ do not depend on $\kappa,\delta,N$.
\end{prop}
\begin{proof}
Let us decompose $L^{\rho_-}(q)=L^{\rho_-}(N,N)+[L^{\rho_-}(q)-L^{\rho_-}(N,N)]$. Define the random variables
\begin{equation}
\begin{aligned}
A_N&=\frac{L^{\rho_-}(N,N)-4N-2^{4/3}N^{1/3}(\kappa-\delta)^2}{2^{4/3}N^{1/3}},\\
B_N&=\frac{L^{\rho_-}(q)-L^{\rho_-}(N,N)+2^{4/3}N^{1/3}2(\kappa-\delta)u_2}{2^{4/3}N^{1/3}}.
\end{aligned}
\end{equation}
Then
\begin{equation}
\begin{aligned}
\Pb(L^{\rho_-}(q)<S)&=\Pb(A_N+B_N\leq -\tfrac12(\kappa-\delta)^2)\\
&\leq \Pb(A_N\leq -\tfrac14(\kappa-\delta)^2)+\Pb(B_N\leq -\tfrac14(\kappa-\delta)^2).
\end{aligned}
\end{equation}
Since we are in the stationary situation we know (see Lemma~2.1 of~\cite{BFO20}) that
\begin{equation}
L^{\rho_-}(q)-L^{\rho_-}(N,N)=\sum_{k=1}^{u_2(2N)^{2/3}} Z_k,
\end{equation}
with $Z_1,Z_2,\ldots$ independent random variables distributed as $Z_k\sim {\rm Exp}(1-\rho_-)-{\rm Exp}(\rho_-)$ (where the two exponential distributions are independent). A simple computation gives
\begin{equation}
\E(B_N)=\frac{u_2(2N)^{2/3} ((1-\rho_-)^{-1}-\rho_-^{-1})+2^{4/3}N^{1/3}2u_2(\kappa-\delta)}{2^{4/3}N^{1/3}}=\Or(N^{-2/3}).
\end{equation}
Using standard exponential Chebyshev inequality, see Lemma~\ref{LemmaBoundRandomWalks} with $L=u_2$, $\varkappa=\delta-\kappa$ and $\xi=(\kappa-\delta)^2/4$, we obtain
\begin{equation}\label{eq1.28}
\Pb(B_N\leq -\tfrac14(\kappa-\delta)^2)\leq C e^{-c (\kappa-\delta)^4/u_2}
\end{equation}
for all $N$ large enough, where the constants constants $C,c$ are independent of $N$.

It remains to bound $\Pb(A_N\leq -\tfrac14(\kappa-\delta)^2)$. Notice that the scaled random variable in $A_N$ is $L^{\rho_-}(N,N)$, for which we know that $L^{\rho_-}(N,N)\geq \tilde L^{\rm pp}(N,N)$ where $\tilde L^{\rm pp}$ has parameter $\rho_-$ on the diagonal, that is, it has weights as in \eqref{eqWeightsPP} with parameter $\alpha=(\delta-\kappa) 2^{-4/3}N^{-1/3}$. Defining
\begin{equation}
\widetilde A_N=\frac{\tilde L^{\rm pp}(N,N)-4N-2^{4/3}N^{1/3}(\kappa-\delta)^2}{2^{4/3}N^{1/3}}
\end{equation}
we get
\begin{equation}
\Pb(A_N\leq -\tfrac14(\kappa-\delta)^2)\leq \Pb(\widetilde A_N\leq -\tfrac14(\kappa-\delta)^2).
\end{equation}
We know that $\widetilde A_N$ converges to a non-trivial distribution in the $N\to\infty$ limit\footnote{In~\cite{BR99}, Theorem~4.2(iii), Baik and Rains obtained the limit to the distribution function $F^{\squareslash}$ for the geometric LPP model, whose limiting distribution is given in terms of a Riemann-Hilbert problem. In~\cite{Bai02}, Section~7, the expected limiting result for the exponential model was stated, but details in terms of Riemann-Hilbert problem has not been written down, although there is no doubt that it works. However, we also know that for the geometric model, the limiting distribution can be written as a Fredholm Pfaffian, see~\cite{BBNV18}, Theorem~4.2. The analogue result in term of Pfaffians for the exponential model is the one-point case of~\cite{BBCS17}, Theorem~1.7. Combining these results we have that for the exponential model the limiting distribution is indeed $F^{\squareslash}$.}:
\begin{equation}\label{eq1.20}
\lim_{N\to\infty} \Pb(\widetilde A_N\leq -\tfrac14(\kappa-\delta)^2) = F^{\squareslash}(3 w^2;w)
\end{equation}
with $w=-\frac12(\kappa-\delta)$. Now we apply Theorem~\ref{thmRHP} with $\mu=3$: there exists a constant $C$ such that for all $-o(N^{1/12})<w<0$,
\begin{equation}\label{eq1.33}
\Pb(\widetilde A_N\leq -\tfrac14(\kappa-\delta)^2) \leq C e^{c w^3}
\end{equation}
with $c=2\sqrt{3}-\frac{10}{3}>0$ uniformly for all $N$ large enough. Combining \eqref{eq1.28} and \eqref{eq1.33}, the claimed result is proven.
\end{proof}

The next LPP to be analyzed is $L^{\rho_-}_{{\D}^c}(q)$ and it is given by the random variables
\begin{equation}
\left\{
\begin{aligned}
\omega^{\D^c}_{0,0} &= 0, &\\
\omega^{\D^c}_{i,i} &=0,&\\
\omega^{\D^c}_{i,0} &\sim \exp(1-\rho_-), &\text{ for }i\in\N, \\
\omega^{\D^c}_{i,j} &\sim \exp(1), &\text{ for }(i,j)\in {\cal B},\\
\end{aligned}
\right.
\end{equation}
which is the setting illustrated in Figure~\ref{fig:ptp_lpp}.
\begin{figure}[t!]
  \centering
   \includegraphics[height=5cm]{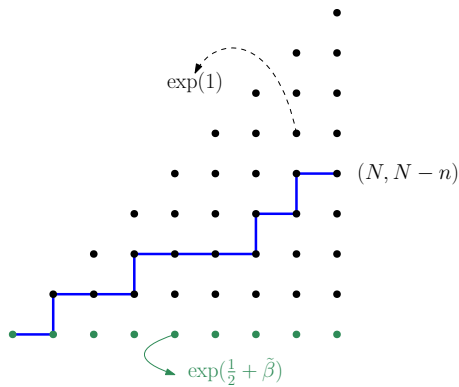}
   \caption{A point-to-point LPP path from $(1, 1)$ to $(N, N-n)$. This corresponds to the integrable case of Figure 2 and Theorem 3.1 of \cite{BFO20} with $\tilde \alpha=1/2$ and $\tilde \beta=(\kappa-\delta)2^{-4/3} N^{-1/3}$.}
   \label{fig:ptp_lpp}
\end{figure}

This LPP model has a kernel given in Theorem~3.1 of~\cite{BFO20} (we put\, $\tilde{}$\, to the parameters in~\cite{BFO20} to avoid misunderstanding) with the mapping of the coordinates $q+(0,1)=(\tilde N,\tilde N-\tilde n)$ and parameters $\tilde \alpha=\frac12$, $\tilde \beta=\frac12-\rho_-=(\kappa-\delta)2^{-4/3}N^{-1/3}$. The shift by $(0,1)$ is due to the fact that in~\cite{BFO20} the lowest-left point where a random variable is not $0$ is $(1,1)$, while here is $(1,0)$. Theorem~3.1 of~\cite{BFO20} applies without problems for $\tilde \alpha=\frac12$, although there it was stated for $\tilde \alpha<1/2$. The only real condition was $\tilde \alpha+\tilde \beta>0$, which is satisfied here since $\tilde \alpha+\tilde \beta=1-\rho_->0$. We have
\begin{equation}
\begin{aligned}
\Pb(L^{\rho_-}_{D^c}(q)\leq S)&={\rm Pf}(J-K)_{L^2(S,\infty)\times L^2(S,\infty)}\\
&=\sum_{m\geq 0} \frac{(-1)^{m}}{m!} \int_S^\infty dx_1\cdots \int_S^\infty dx_m {\rm Pf}[K(x_i,x_j)]_{1\leq i,j\leq m},
\end{aligned}
\end{equation}
with $J=\left(\begin{array}{cc} 0 & 1 \\ -1 & 0 \\ \end{array}\right)$ and $K$ is the $2\times 2$ matrix kernel given by
\begin{equation}
  \begin{aligned} \label{eq:kernel}
   K_{11}(x,y) =& - \oint \frac{dz}{2\pi \I} \oint\frac{dw}{2\pi \I}\frac{\Phi(x,z)}{\Phi(y,w)}\left[(\tfrac 12-z)(\tfrac 12+w)\right]^{\tilde n} \frac{(z+\tfrac 12)(w-\tfrac 12)(z+w)}{4zw}R(z,w,\tilde\beta), \\
   K_{12}(x,y) =& - \oint \frac{dz}{2\pi\I} \oint \frac{dw}{2\pi\I}\frac{\Phi(x,z)}{\Phi(y,w)}\left[\frac{\tfrac 12-z}{\tfrac 12-w}\right]^{\tilde n} \frac{z+\tfrac 12}{w+\tfrac 12}\frac{z+w}{2z}R(z,w,\tilde\beta) \\
               =& -K_{21} (y,x),\\
   K_{22}(x,y) =& \oint \frac{dz}{2\pi\I} \oint \frac{dw}{2\pi\I} \frac{\Phi(x,z)}{\Phi(y,w)} \frac{1}{\left[(\tfrac 12+z)(\tfrac 12-w)\right]^{\tilde n}}\frac{z+w}{(z-\tfrac 12)(w+\tfrac 12)}R(z,w,\tilde\beta) + \tilde \varepsilon(x, y),
  \end{aligned}
  \end{equation}
  where we denoted
  \begin{equation}
R(z,w,\tilde\beta)=\frac{(z+\tilde \beta)(w-\tilde \beta)}{(z-\tilde \beta)(w+\tilde \beta)}\frac{1}{z-w}.
  \end{equation}
The integration contours are for all cases $(z,w) \in \Gamma_{1/2,\tilde\beta}\times\Gamma_{-1/2,-\tilde\beta}$. In~\cite{BFO20} we gave the contours by removing some zero contributions, which are the cases of $(z,w)\in\Gamma_{\tilde\beta}\times\Gamma_{-\tilde\beta}$. Since in our situation we have $\tilde\beta>0$, we can keep them inside the contours.

We recall the notation used in \cite{BFO20}, where
\begin{equation} \label{eq:phi}
  \Phi(x, z)= e^{-xz} \phi (z) \quad \textrm{with}\quad \phi(z)= \left[ \frac{ \tfrac12 + z } { \tfrac12-z }\right]^{\tilde N-1}
\end{equation}
and
\begin{equation}
 \tilde\varepsilon(x, y) =-\sgn(x-y)\oint\limits_{\Gamma_{1/2}} \frac{dz}{2\pi\I} \frac{2z e^{-z |x-y|}}{\left( \frac{1}{4} - z^2 \right)^{\tilde n+1}}.
\end{equation}

\begin{prop}\label{PropUpperBoundNotDiagonal}
Let $S$ be chosen as in \eqref{eqChoiceOfS}, and $\alpha$, $\rho_-$ as in Proposition~\ref{PropCrossing}. Take $\kappa-\delta\geq \max\{6u_2,1\}$ and $u_2>0$. Then uniformly for all $N$ large enough we have
\begin{equation}
\Pb(L^{\rho_-}_{\D^c}(q)\geq S) \leq C e^{-(\kappa-\delta)^3/24}.
\end{equation}
\end{prop}
\begin{proof}
We have
\begin{equation}
\tilde N=N+u_2 (2N)^{2/3},\quad \tilde n=2u_2 (2N)^{2/3}-1.
\end{equation}
Recall that $S=4N+2^{4/3}N^{1/3}[\tfrac12 (\kappa-\delta)^2-2u_2 (\kappa-\delta)]$ and let us set
\begin{equation}\label{eq1.48}
x_i=S+2^{4/3}N^{1/3}\xi_i.
\end{equation}
Then
\begin{equation}
\Pb(L^{\rho_-}_{D^c}(q)> S)\leq \sum_{m\geq 1} \frac{1}{m!} \int_0^\infty d\xi_1\cdots \int_0^\infty d\xi_m |{\rm Pf}[\widetilde K(\xi_i,\xi_j)]_{1\leq i,j\leq m}|,
\end{equation}
where the rescaled kernel is given by
\begin{equation}
\widetilde K(\xi_i,\xi_j)=
\left(
  \begin{array}{cc}
    e^{f(x_i)+f(x_j)} K_{1,1}(x_i,x_j) & 2^{4/3}N^{1/3} e^{f(x_i)-f(x_j)} K_{1,2}(x_i,x_j) \\
    2^{4/3}N^{1/3} e^{f(x_j)-f(x_i)} K_{2,1}(x_i,x_j) & (2^{4/3}N^{1/3})^2 e^{-f(x_i)-f(x_j)} K_{2,2}(x_i,x_j) \\
  \end{array}
\right).
\end{equation}
The terms $e^{f(x_i)}$ are conjugation factors which do not change the value of the Pfaffian. In our case we can take
\begin{equation}
f(x_i)=-N^{2/3}f_1(1/2-R)+\tfrac{1}{24}(\kappa-\delta)^3+\tfrac14\xi_i(\kappa-\delta),
\end{equation}
with $f_1$ and $R$ given in the proof of Lemma~\ref{LemBoundKernel} below.

In Lemma~\ref{LemBoundKernel} we show that, for any given $S\in\R$, there exist constants $C>0$ such that for all $\xi_i,\xi_j\geq S$ we have
\begin{equation}
|\widetilde K(\xi_i,\xi_j)_{k,\ell}|\leq C e^{-c (\xi_i+\xi_j)-(\kappa-\delta)^3/24},
\end{equation}
with $c=(\kappa-\delta)/4$, for all $N$ large enough. This and Hadamard bound imply that
\begin{equation}
\Pb(L^{\rho_-}_{D^c}(q)> S)\leq \sum_{m\geq 1}\frac{(2m)^{m/2}}{m!} C^m e^{-m (\kappa-\delta)^3/24} \bigg(\int_0^\infty e^{-2 c\xi} d\xi\bigg)^m \leq \tilde C e^{-(\kappa-\delta)^3/24}
\end{equation}
for some new constant $\tilde C$.
\end{proof}

\begin{lem}\label{LemBoundKernel}
Let us assume that $\kappa-\delta\geq \max\{6 u_2,1\}$ and $u_2>0$. Recall the scaling for $x_i$ in \eqref{eq1.48}. Then there exists a constant $C$ such that for all $N$ large enough
\begin{equation}
\begin{aligned}
|e^{f(x_1)+f(x_2)} K_{11}(x_1,x_2)|&\leq C e^{-(\kappa-\delta)^3/24} e^{-(\xi_1+\xi_2) (\kappa-\delta)/4},\\
|2^{4/3}N^{1/3} e^{f(x_1)-f(x_2)} K_{12}(x_1,x_2)|&\leq C e^{-(\kappa-\delta)^3/24} e^{-(\xi_1+\xi_2) (\kappa-\delta)/4},\\
|(2^{4/3}N^{1/3})^2 e^{-f(x_2)-f(x_2)}K_{22}(x_1,x_2)|&\leq C e^{-(\kappa-\delta)^3/24} e^{-(\xi_1+\xi_2) (\kappa-\delta)/4},
\end{aligned}
\end{equation}
for all $\xi_1,\xi_2\geq 0$.
\end{lem}
\begin{proof}
Let us start deriving the bound for $K_{11}$ without the conjugation and prefactor, which can be added at the end.
Let us recall that
\begin{equation}
K_{11}(x_1,x_2) = \oint \frac{dz}{2\pi \I} \oint\frac{dw}{2\pi \I}\frac{\Phi(x_1,z)}{\Phi(x_2,w)}\left[(\tfrac 12-z)(\tfrac 12+w)\right]^{\tilde n} \frac{(z+\tfrac 12)(\tfrac 12-w)(z+w)}{4zw}R(z,w,\tilde\beta).
\end{equation}

Let us denote $\mu_i=\tfrac12(\kappa-\delta)^2-2 u_2 (\kappa-\delta)+\xi_i\geq 0$ since we assumed $\kappa-\delta\geq 6 u_2$ and we have $\xi_i\geq 0$. Then
\begin{equation}
K_{11}(x_1,x_2) = \oint \frac{dz}{2\pi \I} \oint\frac{dw}{2\pi \I} \frac{e^{N f_0(z)+N^{2/3} f_1(z)+N^{1/3} f_2(z,\mu_1)}}{e^{N f_0(w)-N^{2/3} f_1(w)+N^{1/3} f_2(w,\mu_2)}} \frac{(z+w)}{4zw}R(z,w,\tilde\beta)
\end{equation}
with
\begin{equation}\label{eq1.55}
\begin{aligned}
f_0(z)&=-4 z+\ln(\tfrac12+z)-\ln(\tfrac12-z)=-f_0(-z),\\
f_1(z)&=2^{2/3}u_2 (\ln(\tfrac12+z)+\ln(\tfrac12-z))=f_1(-z),\\
f_2(z,\mu)&=-2^{4/3} \mu z=-f_2(-z,\mu).
\end{aligned}
\end{equation}
Let us choose the integration contours as follows:
\begin{equation}\label{eqchoiceOfContours}
z=\tfrac12-R e^{\I \phi},\quad w=-\tfrac12+R e^{\I\theta},\quad R=\tfrac12(1-\tilde\beta),
\end{equation}
with $\phi,\theta\in (-\pi,\pi]$.

First we bound the terms not written in the exponential form. We have
\begin{equation}\label{eq1.57}
\left|\frac{z+w}{4zw}\right|=\left|\frac{1}{4z}+\frac{1}{4w}\right|\leq \frac{1}{1-2R}=\frac{1}{\tilde\beta}
\end{equation}
and
\begin{equation}\label{eq1.58}
|R(z,w,\tilde\beta)|= \left|\frac{1}{z-w}+\frac{2\tilde\beta}{(z-\tilde\beta)(w+\tilde\beta)}\right|\leq \frac{1}{1-2R}+\frac{2\tilde\beta}{(1/2-R)^2}=\frac{9}{\tilde\beta}.
\end{equation}

Next consider the exponential terms. The value at $\phi=\theta=0$ is given by
\begin{equation}
\begin{aligned}
M_N&=\frac{e^{N f_0(1/2-R)+N^{2/3} f_1(1/2-R)+N^{1/3} f_2(1/2-R,\mu_1)}}{e^{N f_0(-1/2+R)-N^{2/3} f_1(-1/2+R)+N^{1/3} f_2(-1/2+R,\mu_2)}}\\
&= e^{2N f_0(1/2-R)+2N^{2/3}f_1(1/2-R)+N^{1/3} f_2(1/2-R,\mu_1+\mu_2)},
\end{aligned}
\end{equation}
where we used the symmetries properties in \eqref{eq1.55} as well as the linearity in $\mu$ of $f_2(z,\mu)$.
We have
\begin{equation}
\begin{aligned}
\Re(f_0(z)) &= -2 +4R\cos(\phi) +\tfrac12 \ln(1+R^2-2R \cos(\phi)) -\ln(R),\\
\frac{d\Re(f_0(z))}{d\phi} &= -R\sin(\phi) \left[4-\frac{1}{|\tfrac12+z|^2}\right]<0
\end{aligned}
\end{equation}
for all $\phi\in (0,\pi)$ and $0<R\leq \frac12$. Similarly for $-\Re(f_0(w))$. Thus the contours in \eqref{eqchoiceOfContours} are steep descent for $f_0(z)$ and $-f_0(w)$ respectively. For large values of $\xi_i$ (with $\xi_i=\Or(N^{2/3})$), there are other terms that are in the exponential scale in $N$, namely
\begin{equation}
\frac{e^{N^{1/3} \Re(f_2(z,\mu_1))}}{e^{N^{1/3}\Re(f_2(w,\mu_2))}} = \frac{e^{-2^{4/3}N^{1/3}\mu_1 (1/2-R\cos(\phi))}}{e^{-2^{4/3}N^{1/3}\mu_2 (-1/2+R\cos(\theta))}}\leq \frac{e^{-2^{4/3}N^{1/3}\mu_1 (1/2-R)}}{e^{-2^{4/3}N^{1/3}\mu_2 (-1/2+R)}}
\end{equation}
provided $\mu_1,\mu_2\geq 0$. This is the case when $\xi_1,\xi_2$ is larger than a fixed constant (depending on $u_2$, which is however fixed). This implies that for any given (small) $\e>0$, for $|\phi|\geq \e$ or $|\theta|\geq \e$ we have
\begin{equation}\label{eq1.62}
\left|\frac{e^{N f_0(z)+N^{2/3} f_1(z)+N^{1/3} f_2(z,\mu_1)}}{e^{N f_0(w)-N^{2/3} f_1(w)+N^{1/3} f_2(w,\mu_2)}}\right|\leq  C e^{-c(\e) N} M_N
\end{equation}
for some constant $c=c(\e)$.
Furthermore, the contribution coming from $|\phi|\leq\e$ and $|\theta|\leq\e$ can be estimated using Taylor expansions. We have
\begin{equation}
\begin{aligned}
\Re(f_0(z))&=f_0(\tfrac12-R)-\frac{R \phi^2}{2}(4-(1-R)^{-2})-\phi^4+\Or((1-2R)\phi^4)\\
\Re(f_1(z))&=f_1(\tfrac12-R)+\frac{R u_2 \phi^2}{2^{1/3} (1-R)^2}+\Or(\phi^4)
\end{aligned}
\end{equation}
and
\begin{equation}
\begin{aligned}
-\Re(f_0(w))&=-f_0(-\tfrac12+R)-\frac{R \theta^2}{2}(4-(1-R)^{-2})-\theta^4+\Or((1-2R)\theta^4)\\
\Re(f_1(w))&=f_1(-\tfrac12+R)+\frac{R u_2 \theta^2}{2^{1/3} (1-R)^2}+\Or(\theta^4).
\end{aligned}
\end{equation}
Therefore we get, with $R=\tfrac12(1-\tilde\beta)=\tfrac12(1-(\kappa-\delta)2^{-4/3}N^{-1/3})$,
\begin{equation}
\begin{aligned}
\Re(N f_0(z)+N^{2/3} f_1(z))&=N f_0(\tfrac12-R)+N^{2/3} f_1(\tfrac12-R)+N^{1/3} f_2(\tfrac12-R,\mu_1)\\
& -N^{2/3} 2^{-1/3} \phi^2 (\kappa-\delta-2u_2) -N\phi^4+\Or(N^{2/3}\phi^4).
\end{aligned}
\end{equation}
For all $N$ large enough, $-N\phi^4+\Or(N^{2/3}\phi^4)\leq 0$. Thus, for all $\kappa-\delta>6u_2$ and $\e>0$ small enough (taken independently of $N$) we have
\begin{multline}\label{eq1.65}
\Re(N f_0(z)+N^{2/3} f_1(z)+N^{1/3} f_2(z,\mu_1))\\
\leq N f_0(\tfrac12-R)+N^{2/3} f_1(\tfrac12-R)+N^{1/3} f_2(\tfrac12-R,\mu_1)
 -\tfrac13 2^{2/3}\phi^2 (\kappa-\delta)N^{2/3}
\end{multline}
for all $N$ large enough. Similarly,
\begin{multline}\label{eq1.66}
\Re(-N f_0(w)+N^{2/3} f_1(w)-N^{1/3} f_2(w,\mu_2))\\
\leq -N f_0(-\tfrac12+R)+N^{2/3} f_1(\tfrac12-R)-N^{1/3} f_2(\tfrac12-R,\mu_2) -\tfrac13 2^{2/3} \theta^2 (\kappa-\delta)N^{2/3}.
\end{multline}

Using \eqref{eq1.57}, \eqref{eq1.58}, \eqref{eq1.65} and \eqref{eq1.66}, the contribution of $\phi,\theta\in [-\delta,\delta]$ is bounded by
\begin{equation}\label{eq1.68}
M_N \frac{C}{\tilde\beta^2 (\kappa-\delta) N^{2/3}} = M_N \frac{4C}{(\kappa-\delta)^{3}}
\end{equation}
for some constant $C$.

To resume, first we take $\e>0$ small enough so that \eqref{eq1.68} holds and then use \eqref{eq1.62} for that $\e$, which is subleading for all $N$ large enough. Therefore we have obtained, for all $\kappa-\delta\geq \max\{6u_2,1\}$,
\begin{equation}
\left|K_{11}(x_1,x_2)\right|\leq C M_N
\end{equation}
for some new constant $C$.
It remains to determine $M_N$. For large $N$ we have
\begin{equation}
\begin{aligned}
M_N&= e^{2N^{2/3}f_1(1/2-R)} e^{-\tfrac12(\xi_1+\xi_2) (\kappa-\delta)} e^{-\tfrac{5}{12}(\kappa-\delta)^3 } e^{2 u_2 (\kappa-\delta)^2} e^{\Or(N^{-1/3})} \\
&\leq 2 e^{2N^{2/3}f_1(1/2-R)}  e^{-\tfrac12(\xi_1+\xi_2) (\kappa-\delta)} e^{-\tfrac{1}{12}(\kappa-\delta)^3},
\end{aligned}
\end{equation}
where in the last inequality we used $u_2\leq \tfrac16(\kappa-\delta)$. Multiplying by the conjugations we then get
\begin{equation}
|e^{f(x_1)+f(x_2)} K_{11}(x_1,x_2)|\leq C e^{-(\kappa-\delta)^3/24} e^{-(\xi_1+\xi_2) (\kappa-\delta)/4}.
\end{equation}

The estimates for $K_{12}$ and the double integral in $K_{22}$ are similar. We have
\begin{equation}
K_{12}(x_1,x_2) = -\oint \frac{dz}{2\pi \I} \oint\frac{dw}{2\pi \I} \frac{e^{N f_0(z)+N^{2/3} f_1(z)+N^{1/3} f_2(z,\mu_1)}}{e^{N f_0(w)+N^{2/3} f_1(w)+N^{1/3} f_2(w,\mu_2)}} \frac{(z+w)}{2z}R(z,w,\tilde\beta).
\end{equation}
We bound
\begin{equation}
\left|\frac{(z+w)}{2z}\right|\leq \frac{1}{|z|}\leq \frac{2}{\tilde\beta}
\end{equation}
and instead of $M_N$ we have
\begin{equation}
e^{2N f_0(1/2-R)+N^{1/3} f_2(1/2-R,\mu_1+\mu_2)}\leq 2 e^{-\tfrac12(\xi_1+\xi_2) (\kappa-\delta)} e^{-\tfrac{1}{12}(\kappa-\delta)^3}
\end{equation}
for all $N$ large enough. Multiplying with the conjugation we get
\begin{equation}
\begin{aligned}
|2^{4/3}N^{1/3} e^{f(x_1)-f(x_2)} K_{12}(x_1,x_2)|&\leq C  e^{-(\kappa-\delta)^3/12} e^{-\xi_1(\kappa-\delta)/4} e^{-\xi_2(\kappa-\delta)3/4}\\
&\leq C  e^{-(\kappa-\delta)^3/24} e^{-(\xi_1+\xi_2)(\kappa-\delta)/4}.
\end{aligned}
\end{equation}

Similarly,
\begin{equation}
K_{22}(x_1,x_2) =\tilde\e(x_1,x_2)+ \oint \frac{dz}{2\pi \I} \oint\frac{dw}{2\pi \I} \frac{e^{N f_0(z)-N^{2/3} f_1(z)+N^{1/3} f_2(z,\mu_1)}}{e^{N f_0(w)+N^{2/3} f_1(w)+N^{1/3} f_2(w,\mu_2)}} (z+w) R(z,w,\tilde\beta).
\end{equation}
We bound
\begin{equation}
\left|z+w\right|\leq 2
\end{equation}
and instead of $M_N$ we have
\begin{equation}
e^{2N f_0(1/2-R)-2N^{2/3}f_1(1/2-R)+N^{1/3} f_2(1/2-R,\mu_1+\mu_2)}\leq 2e^{-2N^{2/3}f_1(1/2-R)} e^{-\tfrac12(\xi_1+\xi_2) (\kappa-\delta)} e^{-\tfrac{1}{12}(\kappa-\delta)^3}
\end{equation}
for all $N$ large enough. Multiplying by the conjugation $(2^{4/3}N^{1/3})^2 e^{-f(x_2)-f(x_2)}$ factor, the term of $K_{22}$ coming from the double integral is bounded by
\begin{equation}
 C e^{-(\kappa-\delta)^3/6} e^{-(\xi_1+\xi_2) (\kappa-\delta)3/4}\leq C e^{-(\kappa-\delta)^3/24} e^{-(\xi_1+\xi_2)(\kappa-\delta)/4}.
\end{equation}

It remains to bound $\tilde\e(x_1,x_2)$. It is antisymmetric in $(x_1,x_2)$ and, for $x_1-x_2\geq 0$ given by
\begin{equation}
\tilde\varepsilon(x_1,x_2) = -\oint\limits_{\Gamma_{1/2}} \frac{dz}{2\pi\I} \frac{2z e^{-z (x_1-x_2)}}{\left( \frac{1}{4} - z^2 \right)^{\tilde n+1}}
= -\oint\limits_{\Gamma_{1/2}} \frac{dz}{2\pi\I} 2z e^{-2 N^{2/3} f_1(z)} e^{N^{1/3} f_2(z,\xi_1-\xi_2)}.
\end{equation}
The path $z$ chosen above is steep descent for $-f_1(z)$ for any $u_2>0$. Applying the estimates obtained above, after a few computations, we obtain
\begin{equation}
\left|\tilde\varepsilon(x_1,x_2)\right|\leq C e^{-2N^{2/3} f_1(1/2-R)} N^{-2/3} e^{-\tfrac12(\kappa-\delta)|\xi_1-\xi_2|}.
\end{equation}
Thus we get
\begin{equation}
\begin{aligned}
\left|(2^{4/3}N^{1/3})^2 e^{-f(x_2)-f(x_2)} \tilde\varepsilon(x_1,x_2)\right|&\leq  C e^{-(\kappa-\delta)^3/12} e^{-(\xi_1+\xi_2) (\kappa-\delta)/4} e^{-|\xi_1-\xi_2|/2}\\
&\leq C e^{-(\kappa-\delta)^3/24} e^{-(\xi_1+\xi_2) (\kappa-\delta)/4}.
\end{aligned}
\end{equation}
\end{proof}

\section{Weak convergence of the point-to-point process}\label{SectWeakConvergence}
In this section we prove that the scaled point-to-point LPP in half-space, whose finite-dimensional distributions have been determined by Baik-Barraquand-Corwin-Suidan in~\cite{BBCS17}, is tight and thus we lift the convergence to weak convergence in the space of continuous functions.

Consider the point-to-point LPP with starting point at the origin, end-point given by
\begin{equation}
Q(u)=(N+u (2N)^{2/3},N-u (2N)^{2/3}),
\end{equation}
and with parameter
\begin{equation}
\alpha=\delta 2^{-4/3} N^{-1/3}.
\end{equation}
Define the rescaled point-to-point LPP to $Q(u)$ by
\begin{equation}\label{eqRescPtPt}
{\cal L}^{\rm pp}_N(u):=\frac{L^{\rm pp}(Q(u))-4N}{2^{4/3} N^{1/3}}.
\end{equation}
It is proven in Theorem~1.7 of~\cite{BBCS17} that
\begin{equation}\label{eqPPprocess}
 \lim_{N\to\infty} {\cal L}^{\rm pp}_N(u) = {\cal A}^{\rm pp}_\delta(u)-u^2
\end{equation}
in the sense of finite-dimensional distributions. The limit process ${\cal A}^{\rm pp}_\delta$ has joint distributions given by a Fredholm Pfaffian: for any $0\leq u_1<u_2<\ldots<u_n$, one has
\begin{equation}
 \Pb\bigg(\bigcap_{k=1}^n\{{\cal A}^{\rm pp}_\delta(u_k)\leq s_k\}\bigg)=\Pf(J-P_s K^{\rm pp}_\delta)_{L^2(\R\times\{u_1,\ldots,u_n\})}
\end{equation}
where $P_s(x,u_k)=\Id_{[x\geq s_k]}$ and the $2\times 2$ crossover kernel $K^{\rm pp}_\delta$ is given in Section~2.5 of~\cite{BBCS17} (replace $\eta_k$ with $u_k$ and $\varpi$ with $\delta$ in their formulas).

\begin{remark}\label{RemSlowDec}
If instead of taking the end-point $Q(u)$ we take the end-point on a horizontal line, namely $\tilde Q(u)=(N+2u(2n)^{2/3},N)$, then
\begin{equation}
\tilde {\cal L}^{\rm pp}_N(u):=\frac{L^{\rm pp}(\tilde Q(u))-4N-4u(2N)^{2/3}}{2^{4/3} N^{1/3}}\stackrel{N\to\infty}{\longrightarrow} {\cal A}^{\rm pp}_\delta(u)-u^2
\end{equation}
in the sense of finite-dimensional distribution. This is a consequence of \eqref{eqPPprocess} together with slow-decorrelation, see Theorem~2.1 of~\cite{CFP10b}. This phenomenon implies that all cuts in the LPP, have the same fluctuations, except for the ones corresponding to the characteristic directions, see also~\cite{BFP09} for a further example.
\end{remark}

\begin{thm}\label{thmWeakCvg}
Consider the rescaled process $u\mapsto {\cal L}^{\rm pp}_N(u)$ as defined in \eqref{eqRescPtPt}. Then,
\begin{equation}
\lim_{N\to\infty} {\cal L}^{\rm pp}_N(u) = {\cal A}^{\rm pp}_\delta(u)-u^2
\end{equation}
in the sense of weak convergence on the space of continuous functions on compact intervals.
\end{thm}
\begin{proof}
The result follows from the finite-dimensional convergence \eqref{eqPPprocess} together with tightness shown in Proposition~\ref{propTightnessPointToPoint} below.
\end{proof}

\begin{rem}\label{RemFctSlowDec}
The same result holds for $\tilde {\cal L}^{\rm pp}_N(u)$, namely
\begin{equation}
\lim_{N\to\infty} \tilde {\cal L}^{\rm pp}_N(u) = {\cal A}^{\rm pp}_\delta(u)-u^2
\end{equation}
in the sense of weak convergence on the space of continuous functions on compact intervals. This can be obtained in a similar way by considering a different section in the LPP, but it also is a consequence of a functional slow-decorrelation first derived by Corwin-Liu-Wang in Theorem~2.15 of~\cite{CLW16} is the geometric LPP model, see also Theorem~2.10 of~\cite{CFS16} for an example in the exponential LPP case.
\end{rem}

\begin{prop}\label{propTightnessPointToPoint}
The rescaled process $u\mapsto {\cal L}^{\rm pp}_N(u)$ is tight in the space of continuous functions on a bounded interval.
\end{prop}
\begin{proof}
The proof is by now quite standard and therefore let us indicate the main steps along the lines of~\cite{CP15b} without writing all the details. First of all, by the lower and upper tail estimates given in Appendix~\ref{AppRoughBounds}, we know that
for all $\e>0$ there exists an $S$ such that
\begin{equation}
\Pb(|{\cal L}^{\rm pp}_N(0)|\geq S)\leq \e
\end{equation}
for all $N$ large enough. Thus the random variable ${\cal L}^{\rm pp}_N(0)$ is tight. To show tightness of the process $u\mapsto {\cal L}^{\rm pp}_N(u)$ in the space of continuous functions on a bounded interval, we need to control the modulus of continuity: let
\begin{equation}
\omega_N(\tilde\delta)=\sup_{\stackrel{0\leq u_1,u_2\leq \tilde\delta:}{|u_2-u_1|\leq \tilde\delta}}|{\cal L}^{\rm pp}_N(u_1)-{\cal L}^{\rm pp}_N(u_2)|.
\end{equation}
By Theorem~8.2 of~\cite{Bil68} we need to prove that for any $\e,\tilde\e>0$, there exists a $\tilde\delta>0$ and a $N_0$ such that
\begin{equation}\label{eq1.46}
\Pb(\omega_N(\tilde\delta)\geq \e)\leq \tilde\e,
\end{equation}
for all $N\geq N_0$.

For the stationary process, consider the same rescaling as \eqref{eqRescPtPt}, namely
\begin{equation}
 {\cal L}^{\rm \rho}_N(u):=\frac{L^{\rho}(Q(u))-4N}{2^{4/3} N^{1/3}}.
\end{equation}
Let $\rho_+=\rho=\tfrac12+\delta 2^{-4/3} N^{-1/3}$ and $\rho_-=\tfrac12+(\delta-\kappa) 2^{-4/3}N^{-1/3}$. Assume that $\kappa-\delta>6 M$. Then, by Corollary~\ref{CorCrossingA} and Proposition~\ref{PropCrossing}, for all $0\leq u_1<u_2\leq M$,
\begin{equation}\label{eq3.9}
\begin{aligned}
{\cal L}^{\rm \rho_-}_N(Q(u_2))-{\cal L}^{\rm \rho_-}_N(Q(u_1))&\leq {\cal L}^{\rm  pp}_N(Q(u_2))-{\cal L}^{\rm pp}_N(Q(u_1))\\
&\leq {\cal L}^{\rm \rho}_N(Q(u_2))-{\cal L}^{\rm \rho}_N(Q(u_1))
\end{aligned}
\end{equation}
on a set $\Omega_{\rm cross}$ with $\Pb(\Omega_{\rm cross})\geq 1-C e^{-c(\kappa-\delta)^3}$ for all $N$ large enough.

Thus, for any $\e>0$ and $N$ large enough,
\begin{equation}
\Pb(\omega_N(\tilde\delta)\geq \e)\leq \Pb(\Omega_{\rm cross}^c)+\Pb(\{\omega_N(\tilde\delta)\geq \e\}\cap \Omega_{\rm cross}).
\end{equation}
For fixed $\tilde\e>0$, choose $\kappa$ large enough such that $\Pb(\Omega_{\rm cross}^c)\leq C e^{-c (\kappa-\delta)^3}\leq \tilde\e/2$.

For the second term, notice that on $\Omega_{\rm cross}$ the inequalities \eqref{eq3.9} hold and therefore
\begin{equation}
\begin{aligned}
 |{\cal L}^{\rm  pp}_N(Q(u_2))-{\cal L}^{\rm pp}_N(Q(u_1))|&\leq |{\cal L}^{\rm \rho_-}_N(Q(u_2))-{\cal L}^{\rm \rho_-}_N(Q(u_1))|\\
&+| {\cal L}^{\rm \rho}_N(Q(u_2))-{\cal L}^{\rm \rho}_N(Q(u_1))|
\end{aligned}
\end{equation}
It is easy to see that, for $v_-=2(\delta-\kappa)$ and $v_+=2 \delta$,
\begin{equation}
\lim_{N\to\infty} {\cal L}^{\rm \rho_-}_N(Q(u_2))-{\cal L}^{\rm \rho_-}_N(Q(u_1)) = (u_2-u_1)v_- +\sqrt{2}{\cal B}_-(u_2-u_1)
\end{equation}
and
\begin{equation}
\lim_{N\to\infty} {\cal L}^{\rm \rho}_N(Q(u_2))-{\cal L}^{\rm \rho}_N(Q(u_1)) = (u_2-u_1)v_+ +\sqrt{2}{\cal B}_+(u_2-u_1)
\end{equation}
where ${\cal B}_\pm$ are standard Brownian motions. At this stage, first choose $\tilde\delta$ small enough such that $|v_\pm|\tilde\delta\leq \e/2$. Then, standard computations on increments of independent random variables lead to $\Pb(\{\omega_N(\tilde\delta)\geq \e\}\cap \Omega_{\rm cross})\leq \tilde\e/2$ for all $\tilde\delta$ small enough.
\end{proof}

\section{Localization of the geodesics}\label{SectLocalization}
In this section we consider a geometric aspect of the geodesics. We show that when the end-point is
\begin{equation}
Q=(N+M_1 (2N)^{2/3},N-M_1 (2N)^{2/3}),
\end{equation}
for a fixed $M_1$, then the probability that the geodesics is at distance $M (2N)^{2/3}$ from the diagonal $\cal D$ goes to zero fast in $M$.
The first step consists in a bound on the probability that the geodesic is not localized at time $N/2$ (a special case of Proposition~\ref{propLocalization}), from which following the approach of~\cite{BSS14}, it can be extended to the full time (Theorem~\ref{thmLocalization}).

\begin{rem}
For the half-space LPP, the generic point-to-point LPP where both end-points are not on the diagonal has not been solved. As a consequence we do not have any direct information on the tails of its distribution. However this is a key input in the proof of localizations in previous papers. This is the main difficulty in the proof of localization of the geodesics and to go around it we need to consider modified LPP models.
\end{rem}

For any fixed $\tau\in (0,1)$, consider the point
\begin{equation}
I(u)=(\tau N+u (2N)^{2/3},\tau N-u (2N)^{2/3})
\end{equation}
and the density
\begin{equation}
 \rho=\frac12+\kappa 2^{-4/3} N^{-1/3}.
\end{equation}
We consider the following three coupled LPP problems:
\begin{enumerate}
\item $L^{\rm pp}$: Point-to-point with $\omega_{i,i}\sim \exp(\rho)$, $i\geq 1$, $\omega_{i,0}=0$, $i\geq 1$,
\item $L^{\rm st,\rho}$: Stationary with parameter $\rho$: $\omega_{i,i}\sim \exp(\rho)$ and $\omega_{i,0}\sim \exp(1-\rho)$, $i\geq 1$ and $\omega_{0,0}=0$,
\item $\tilde L$: Like for $L^{\rm st,\rho}$ but with $\omega_{i,j}=0$ for $\{(i,j)| i+j\geq 2\tau N, i-j\leq 2M_1(2N)^{2/3}\}$ (see Figure~\ref{fig:Ltilde}).
\end{enumerate}
 \begin{figure}[t!]
  \centering
   \includegraphics[height=6cm]{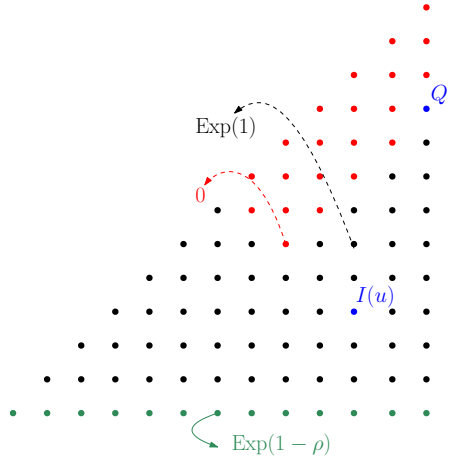}
   \caption{The random environment for the last passage percolation path $\tilde L$. The values in red the region are set to be zero, in green the same boundary variables as in the stationary LPP $L^{\rm st,\rho}$, $\omega_{i,0}\sim \exp(1-\rho)$.}
   \label{fig:Ltilde}
\end{figure}
Denote by $\pi^{\rm pp}$, $\pi^\rho$ and $\tilde \pi$ the geodesics of these LPP.

With the above notations, we have the following localization result.
\begin{thm}\label{thmLocalization}
Let ${\cal L}_M=\{(i,j) | i-j=M (2N)^{2/3}\}$ the line parallel to the diagonal at distance $M(2N)^{2/3}$. For all $M\geq M_1+9$, with $M=\Or(N^{1/3}/\ln(N))$, we have
\begin{equation}
\Pb\left(\pi^{\rm pp}(Q)\cap {\cal L}_M=\emptyset\right) \geq \Pb\left(\pi^{\rho}(Q)\cap {\cal L}_M=\emptyset\right) \geq 1-C e^{-c (M-M_1)^3}
\end{equation}
uniformly for all $N$ large enough.
\end{thm}
\begin{proof}
Once we have the localization bound for the mid-time, i.e., $\tau=1/2$ in Proposition~\ref{propLocalization}, the proof is exactly the one of Theorem~4.4 in~\cite{FB21}, which as mentioned is coming from~\cite{BSS14}, and thus we do not repeat the details.
\end{proof}

\begin{prop}\label{propLocalization}
Assume that $M\geq M_1+9$ with $M=o(N^{1/3})$, $\kappa\leq \frac14(M-M_1)/(1-\tau)$. Then there exist constants $C,c>0$ such that
\begin{equation}
\Pb(\pi^{\rm pp}(Q)\prec I(M))\geq \Pb(\pi^{\rho}(Q)\prec I(M))\geq 1-C e^{-c (M-M_1)^3/(1-\tau)^2}
\end{equation}
uniformly for all $N$ large enough.
\end{prop}

\begin{remark}
At first we though that we could use the approach of comparing the increments with the stationary case as in \cite{FB21}, but in the half-space geometry the situation is more difficult. In the full-space case, the geodesics of the stationary model either uses the randomness on the $x$-axis or in the $y$-axis, but not both simultaneously. In the half-space case, however, the geodesic can both use the randomness on the $x$-axis and the ones on the diagonal. As a consequence, it is not straightforward to prove that the geodesic for the stationary case does not touch the diagonal for a point far from the diagonal, and this is precisely what we would need to apply Proposition~\ref{PropComparison} (b) in order to apply the approach of \cite{FB21}.
\end{remark}

\begin{proof}[Proof of Proposition~\ref{propLocalization}]
In this proof, to keep the formulas slightly smaller we write $L_{a,b}$ instead of $L(a,b)$ and $L_a$ instead of $L(a)$.
By the order of geodesics, we have
\begin{equation}\label{eq2.3}
 \pi^{\rm pp}(Q)\prec \pi^\rho(Q)\prec \tilde \pi(Q).
\end{equation}
Thus to prove that for the point-to-point case as well as for the stationary LPP the geodesic at time $\tau N$ is localized, it is enough to prove that the geodesic $\tilde\pi(Q)$ is localized.

Let us take $M>M_1$. Then we have
\begin{equation}\label{eqLoc1}
\begin{aligned}
\Pb(\tilde\pi(Q)\prec I(M)) &=\Pb\left( \sup_{0\leq u<M} (L^{\rm st,\rho}_{I(u)}+\tilde L_{I(u),Q}) > \sup_{u\geq M}(L^{\rm st,\rho}_{I(u)}+\tilde L_{I(u),Q})\right)\\
&\geq \Pb\left(  L^{\rm st,\rho}_{I(M_1)}+\tilde L_{I(M_1),Q} > \sup_{u\geq M}(L^{\rm st,\rho}_{I(u)}+\tilde L_{I(u),Q})\right).
\end{aligned}
\end{equation}
By using $\tilde L_{I(u),Q}\leq L^{\square}_{I(u),Q}$, where $L^\square$ is the full-space LPP for which $\omega_{i,j}\sim\exp(1)$ for all $(i,j)\in\Z^2$, we get
\begin{equation}\label{eqLoc2}
\begin{aligned}
\eqref{eqLoc1} &\geq \Pb\left(  L^{\rm st,\rho}_{I(M_1)}+\tilde L_{I(M_1),Q} > \sup_{u\geq M}(L^{\rm st,\rho}_{I(u)}+L^\square_{I(u),Q})\right)\\
&=\Pb\left( \sup_{u\geq M}(L^{\rm st,\rho}_{I(u)}- L^{\rm st,\rho}_{I(M_1)}+L^\square_{I(u),Q}-\tilde L_{I(M_1),Q}) <0\right)\\
&=\Pb\left( \sup_{u\geq M}(L^{\rm st,\rho}_{I(u)}- L^{\rm st,\rho}_{I(M)}+L^\square_{I(u),Q}-L^\square_{I(M),Q}+\Delta L^{\rm st,\rho}+\Delta \tilde L) <0\right),
\end{aligned}
\end{equation}
where $\Delta L^{\rm st,\rho}= L^{\rm st,\rho}_{I(M)}-L^{\rm st,\rho}_{I(M_1)}$ and $\Delta \tilde L=L^\square_{I(M),Q}-\tilde L_{I(M_1),Q}$.
Furthermore, for any choice of $A\in\R$,
\begin{equation}\label{eqLoc3}
\begin{aligned}
\eqref{eqLoc2} &\geq \Pb\left( \sup_{u\geq M}(L^{\rm st,\rho}_{I(u)}- L^{\rm st,\rho}_{I(M)}+L^\square_{I(u),Q}-L^\square_{I(M),Q}) < A \leq -\Delta L^{\rm st,\rho}-\Delta \tilde L)\right)\\
&\geq 1-\Pb\left( \sup_{u\geq M}(L^{\rm st,\rho}_{I(u)}- L^{\rm st,\rho}_{I(M)}+L^\square_{I(u),Q}-L^\square_{I(M),Q})\geq A\right) - \Pb\left(A > -\Delta L^{\rm st,\rho}-\Delta \tilde L)\right).
\end{aligned}
\end{equation}
Thus we have to see that for an appropriate choice of $A$, the last two terms are bounded by a function of $M$ which goes to $0$ as $M\to\infty$.

Let us choose
\begin{equation}\label{eqChoiceOfA}
A=\frac14 \frac{(M-M_1)^2}{1-\tau}2^{4/3}N^{1/3}-\frac14\kappa(M-M_1) 2^{4/3} N^{1/3}.
\end{equation}
Then, for the last term in \eqref{eqLoc3} we have
\begin{equation}\label{eqLoc4}
\begin{aligned}
\Pb\left(A > -\Delta L^{\rm st,\rho}-\Delta \tilde L\right) &= \Pb\left(L^{\rm st,\rho}_{I(M_1)}-L^{\rm st,\rho}_{I(M)}+\tilde L_{I(M_1),Q}-L^\square_{I(M),Q} < A\right)\\
&\leq \Pb\left(L^{\rm st,\rho}_{I(M_1)}-L^{\rm st,\rho}_{I(M)} < -\frac14\kappa(M-M_1)2^{4/3}-\frac14 \frac{(M-M_1)^2}{1-\tau}2^{4/3}N^{1/3}\right)\\
&+\Pb\left(\tilde L_{I(M_1),Q}\leq 4(1-\tau)N-\frac14\frac{(M-M_1)^2}{1-\tau}2^{4/3}N^{1/3}\right)\\
&+\Pb\left(-L^\square_{I(M),Q}\leq -4(1-\tau)N+\frac34\frac{(M-M_1)^2}{1-\tau}2^{4/3}N^{1/3}\right).
\end{aligned}
\end{equation}
Combining Lemmas~\ref{lemTerm1}-\ref{lemTerm3} below we get the claimed bound for $\Pb\left(A > -\Delta L^{\rm st,\rho}-\Delta \tilde L\right)$.

To bound the other term in \eqref{eqLoc3} we use comparison with a stationary LPP. We have
\begin{equation}
 L^\square_{I(u),Q}-L^\square_{I(M),Q}\stackrel{d}{=} L^\square_{0,\tilde I(u-M_1)}-L^\square_{0,\tilde I(M-M_1)}
\end{equation}
where $\tilde I(v)=((1-\tau)N+v(2N)^{2/3},(1-\tau)N-v(2N)^{2/3})$. The comparison lemma in full-space (see Lemma~1 of~\cite{CP15b} or Lemma~3.5 of~\cite{FGN17}) gives
\begin{equation}
  L^\square_{0,\tilde I(u-M_1)}-L^\square_{0,\tilde I(M-M_1)}\leq  L^{\tilde\rho}_{0,\tilde I(u-M_1)}-L^{\tilde\rho}_{0,\tilde I(M-M_1)}
\end{equation}
provided the exit point of the geodesic $Z^{\tilde \rho}$ satisfies $Z^{\tilde \rho}_{0,\tilde I(M-M_1)}\geq 0$.
This implies that
\begin{equation}\label{eqLoc5}
\begin{aligned}
& \Pb\left( \sup_{u\geq M}(L^{\rm st,\rho}_{I(u)}- L^{\rm st,\rho}_{I(M)}+L^\square_{I(u),Q}-L^\square_{I(M),Q})\geq A\right)\\
&\leq \Pb\left( \sup_{u\geq M}(L^{\rm st,\rho}_{I(u)}- L^{\rm st,\rho}_{I(M)}+L^{\tilde\rho}_{I(u),Q}-L^{\tilde\rho}_{I(M),Q})\geq A\right)+ \Pb\left(Z^{\tilde \rho}_{0,\tilde I(M-M_1)}< 0\right),
\end{aligned}
\end{equation}
where we recall from \eqref{eqChoiceOfA} that $A=\frac14 \frac{(M-M_1)^2}{1-\tau}2^{4/3}N^{1/3}-\frac14\kappa(M-M_1) 2^{4/3} N^{1/3}$.

We choose the density
\begin{equation}
 \tilde \rho=\frac12-\frac{M-M_1}{2(1-\tau)}2^{-4/3}N^{-1/3}
\end{equation}
and define the direction $\vec\xi=\left(\frac{(1-\tilde\rho)^2}{(1-\tilde\rho)^2+\tilde\rho^2},\frac{\tilde\rho^2}{(1-\tilde\rho)^2+\tilde\rho^2}\right)$ of the characteristics with density $\tilde\rho$. Then, $\tilde I(M-M_1)=\xi \tilde N+r \tilde N^{2/3}$ with $\tilde N=2(1-\tau)N$ and $r=\tfrac12 (M-M_1) (1-\tau)^{-2/3}$. By the result reported in Lemma~4.1, equation (4.5) of~\cite{FB21}, which was proven in Theorem~2.5 and Proposition~2.7 of~\cite{EJS20}, we have
\begin{equation}
\Pb\left(Z^{\tilde \rho}_{0,\tilde I(M-M_1)}< 0\right)\leq e^{-c (M-M_1)^3/(1-\tau)^2}.
\end{equation}

Therefore we are left with bounding the first term in the r.h.s.\ of \eqref{eqLoc5}. We have that $S_v=L^{\rm st,\rho}_{I(M+v)}- L^{\rm st,\rho}_{I(M)}+L^{\tilde\rho}_{I(M+v),Q}-L^{\tilde\rho}_{I(M),Q}$ is a sum of independent random variables $Z_k$, namely
\begin{equation}
S_v=\sum_{k=1}^{v (2N)^{2/3}} Z_k
\end{equation}
where each of the random variable $Z_k$ is itself a linear combination of 4 independent exponential distributed random variables, which we shortly write
\begin{equation}
Z_k\sim \exp(1-\rho)-\exp(\rho)+\exp(1-\tilde\rho)-\exp(\tilde\rho).
\end{equation}
Thus we need to find an upper bound for
\begin{equation}
\Pb\Big(\sup_{v\geq 0} S_v\geq A\Big).
\end{equation}
This is made in Lemma~\ref{lemTerm4} below, whose estimate combined with those of Lemma~\ref{lemTerm1}-\ref{lemTerm3} leads to the claimed result.
\end{proof}

Finally we prove the four bounds used in the proof of Proposition~\ref{propLocalization}.
\begin{lem}\label{lemTerm1}
Assume that $\kappa<\frac1{14}\frac{M-M_1}{1-\tau}$. Then we have
\begin{equation}
\Pb\left(L^{\rm st,\rho}_{I(M_1)}-L^{\rm st,\rho}_{I(M)} < -\left(\frac{\kappa(M-M_1)}{4}+\frac{(M-M_1)^2}{4(1-\tau)}\right)2^{4/3}N^{1/3}\right)\leq C e^{-c (M-M_1)^3/(1-\tau)^2}
\end{equation}
for all $N$ large enough.
\end{lem}
\begin{proof}
We have $L^{\rm st,\rho}_{I(M)}-L^{\rm st,\rho}_{I(M_1)}=\sum_{k=1}^{(M-M_1)(2N)^{2/3}} Z_k$ where $Z_k\sim \exp(1-\rho)-\exp(\rho)$ are independent random variables. Thus
\begin{equation}
\begin{aligned}
&\Pb\left(L^{\rm st,\rho}_{I(M_1)}-L^{\rm st,\rho}_{I(M)} <- \frac14\kappa(M-M_1)2^{4/3} N^{1/3}-\frac14 \frac{(M-M_1)^2}{1-\tau}2^{4/3}N^{1/3}\right) \\
&=\Pb\left(L^{\rm st,\rho}_{I(M)}-L^{\rm st,\rho}_{I(M_1)} > \frac14\kappa(M-M_1)2^{4/3} N^{1/3}+\frac14 \frac{(M-M_1)^2}{1-\tau}2^{4/3}N^{1/3}\right)\\
&=\Pb\left(\frac{L^{\rm st,\rho}_{I(M)}-L^{\rm st,\rho}_{I(M_1)}-2^{4/3} N^{1/3} 2\kappa(M-M_1)}{2^{4/3}N^{1/3}} > \frac14 \frac{(M-M_1)^2}{1-\tau}-\frac74 \kappa (M-M_1)\right).
\end{aligned}
\end{equation}
Lemma~\ref{LemmaBoundRandomWalks} with $L=M-M_1$, $\varkappa=\kappa$ and $\xi=\frac14 \frac{(M-M_1)^2}{1-\tau}-\frac74 \kappa (M-M_1)\geq \frac18 \frac{(M-M_1)^2}{1-\tau}$, gives
\begin{equation}
\Pb\left(L^{\rm st,\rho}_{I(M_1)}-L^{\rm st,\rho}_{I(M)} < -\left(\frac{\kappa(M-M_1)}{4}+\frac{(M-M_1)^2}{4(1-\tau)}\right)2^{4/3}N^{1/3}\right)\leq C e^{-c (M-M_1)^3/(1-\tau)^2}.
\end{equation}
\end{proof}

\begin{lem}\label{lemTerm2}
For all $M\geq M_1$ with $M=o(N^{1/3})$ we have
\begin{equation}
\Pb\left(\tilde L_{I(M_1),Q}\leq 4(1-\tau)N-\frac14\frac{(M-M_1)^2}{1-\tau}2^{4/3}N^{1/3}\right)\leq C e^{-c (M-M_1)^3 /(1-\tau)^4}
\end{equation}
for some constants $C,c>0$.
\end{lem}
\begin{proof}
If we consider the case when the weight on the diagonal is $0$, then the law is given by the Laguerre Symplectic Ensemble for the special case of purely exponential weight, see Corollary~1.4 and the following remark in~\cite{Bai02}. For the LSE, optimal upper bounds for the tails are obtained in Theorem~2 of~\cite{LR10} (see also Theorem~2 of~\cite{BGHK21} for matching lower bounds), see Appendix~\ref{appLaguerreBeta} for more details. We apply Theorem~2 of~\cite{LR10} with $\kappa=n-1/2$, $\beta=4$, $n=(1-\tau)N$ and $\e$ given by the relation
\begin{equation}
(\sqrt{\kappa}+\sqrt{n})^2(1-\e)= 4(1-\tau)N-\frac14\frac{(M-M_1)^2}{1-\tau}2^{4/3}N^{1/3}
\end{equation}
which leads to
\begin{equation}
\e=\frac{1}{4}\frac{2^{-2/3}}{(1-\tau)^2} (M-M_1)^2 N^{-2/3}+\Or(N^{-1}).
\end{equation}
The result is
\begin{equation}
\Pb\left(\tilde L_{I(M_1),Q}\leq 4(1-\tau)N-\frac14\frac{(M-M_1)^2}{1-\tau}2^{4/3}N^{1/3}\right)\leq C e^{-c (M-M_1)^3 /(1-\tau)^4}
\end{equation}
for some constants $C,c>0$ for all $M\ll N^{1/3}$.
\end{proof}

\begin{lem}\label{lemTerm3}
For all $M\geq M_1$ with $M=o(N^{1/3})$, there exist constants $C,c>0$ such that
\begin{equation}
 \Pb\left(L^\square_{I(M),Q}\geq 4(1-\tau)N-\frac34\frac{(M-M_1)^2}{1-\tau}2^{4/3}N^{1/3}\right) \leq C e^{-c (M-M_1)^3/(1-\tau)^2}
\end{equation}
for all $N$ large enough.
\end{lem}
\begin{proof}
$L^\square_{I(M),Q}$ has the same law as the LPP in full-space from the origin to $(m,n)$ with $m=(1-\tau)N+(M-M_1)(2N)^{2/3}$ and $n=(1-\tau)N-(M-M_1)(2N)^{2/3}$. In our case, $\eta=m/n\to 1$ as $N\to\infty$. Then we can use well-known bounds for the full-space LPP, see e.g.\ Theorem~2 of~\cite{LR10} with $\beta=2$: there exist constants $C,c>0$ such that
\begin{equation}
\Pb(L^\square_{0,(m,n)}\geq (\sqrt{m}+\sqrt{n})^2+ s 2^{4/3}n^{1/3})\leq C e^{-c s^{3/2}}
\end{equation}
for all $N$ large enough. In our case the parameter $s$ is given by the relation
\begin{equation}
 (\sqrt{m}+\sqrt{n})^2+ s n^{1/3} = 4(1-\tau)N-\frac34\frac{(M-M_1)^2}{1-\tau}2^{4/3}N^{1/3}
\end{equation}
which gives
\begin{equation}
s=\frac{(M-M_1)^2}{4(1-\tau)^{4/3}}(1+\Or((M-M_1)^2 N^{-2/3}))
\end{equation}
and therefore, for all $M-M_1=o(N^{1/3})$,
\begin{equation}
\Pb(L^\square_{0,(m,n)}\geq (\sqrt{m}+\sqrt{n})^2+ s 2^{4/3}n^{1/3})\leq C e^{-c (M-M_1)^3/(1-\tau)^2}
\end{equation}
for a new constant $c$.
\end{proof}

\begin{lem}\label{lemTerm4}
Assume that $\kappa<\frac14\frac{M-M_1}{1-\tau}$ and $M-M_1\geq 9$. Let $A=\alpha 2^{4/3} N^{1/3}$  with $\alpha=\frac14 \frac{(M-M_1)^2}{1-\tau}-\frac14\kappa(M-M_1)$ and the densities
\begin{equation}
 \rho=\frac12+\kappa 2^{-4/3}N^{-1/3}\quad\textrm{and}\quad \tilde \rho=\frac12-\frac{M-M_1}{2(1-\tau)}2^{-4/3}N^{-1/3}.
\end{equation}
Then, for all $N$ large enough,
\begin{equation}
\Pb\Big(\sup_{v\geq 0} S_v\geq A\Big)\leq C e^{-c (M-M_1)^4/(1-\tau)^2}
\end{equation}
for some constants $C,c>0$.
\end{lem}
\begin{proof}
 We start with decomposing the time into steps of unit size. We will control the position of $S_v$ for $v\in\N$ and the increments between these times separately. We have
\begin{equation}\label{eq2.30}
 \begin{aligned}
 \Pb\Big(\sup_{v\geq 0} S_v<A\Big) & \geq 1-\sum_{m=1}^\infty \Pb\Big(\sup_{v\in [m-1,m]} S_v\geq A\Big)\\
&\geq 1-\sum_{m=1}^\infty \bigg[\Pb\left(S_m\geq A_m\right)+\Pb\Big(\sup_{v\in[m-1,m]} (S_v-S_m)\geq A-A_m\Big)\bigg]
\end{aligned}
\end{equation}
for any choice of $A_1,A_2,\ldots$.

$S_v$ has a negative drift, since
\begin{equation}
\E(S_v)=v(2N)^{2/3}\left(\frac{1}{1-\rho}-\frac1\rho+\frac{1}{1-\tilde\rho}+\frac1{\tilde\rho}\right) = - v \mu  2^{4/3}N^{1/3}
\end{equation}
with $\mu=\frac{M-M_1}{1-\tau}-2\kappa>0$. Therefore we choose
\begin{equation}
\begin{aligned}
A_m&=\tfrac12 A-\tfrac12\mu m 2^{4/3}N^{1/3} = \tfrac12(\alpha-\mu m) 2^{4/3}N^{1/3},\\
A-A_m&=\tfrac12(\alpha+\mu m) 2^{4/3}N^{1/3}.
\end{aligned}
\end{equation}
By the exponential Chebyshev inequality (see Lemma~\ref{LemmaBoundRandomWalks} for similar and more detailed computations)
\begin{equation}\label{eq2.33}
\Pb\left(S_m\geq A_m\right)=\Pb\left(S_m\geq \tfrac12(\alpha-\mu m) 2^{4/3}N^{1/3}\right)\leq \inf_{\lambda>0} \frac{\E\left(e^{\lambda Z_1 2^{-4/3} N^{-1/3}}\right)^{m (2N)^{2/3}}}{e^{\lambda (\alpha-\mu m)/2}}.
\end{equation}
Explicit computations lead to (with $\tilde M=M-M_1$)
\begin{equation}
\eqref{eq2.33} = e^{-\frac1{1024}[12m(\tilde M/(1-\tau)-2\kappa)+\tilde M(\tilde M/(1-\tau)-\kappa)]^2 +\Or(N^{-2/3})}\leq e^{-c_1 m \tilde M^2/(1-\tau)^2} e^{-c_2 \tilde M^4/(1-\tau)^2}
\end{equation}
for some constant $c_1,c_2>0$ (we used the assumption on $\kappa$). From this we get
\begin{equation}
 \sum_{m=1}^\infty \Pb\left(S_m\geq A_m\right)\leq C e^{-c_2 (M-M_1)^4/(1-\tau)^2}
\end{equation}
for some constant $C>0$.

For the other term in \eqref{eq2.30}, we have
\begin{equation}
\Pb\Big(\sup_{v\in[m-1,m]} (S_v-S_m)\geq A-A_m\Big)
=\Pb\Big(\sup_{w\in[0,1]} (S_{m-w}-S_m)\geq A-A_m\Big)
\end{equation}
Define $\tilde S_w=S_{m-w}-S_m$. It is a sum of independent random variables with positive drift, thus a submartingale. For any $\lambda>0$, $e^\lambda \tilde S_w$ is a positive submartingale and therefore
\begin{equation}\label{eq2.37}
\Pb\Big(\sup_{w\in[0,1]}\tilde S_w\geq A-A_m\Big)\leq \inf_{\lambda>0}\frac{\E(e^{\lambda \tilde S_w})}{e^{\lambda (A-A_m)}}=\inf_{\lambda>0}\frac{\E(e^{-\lambda Z_1 2^{-4/3} N^{-1/3}})^{(2N)^{2/3}}}{e^{\lambda (\alpha+\mu m)/2}}.
\end{equation}
Explicit computations give
\begin{equation}
\eqref{eq2.37}= e^{-\frac{1}{1024}[4m(\tilde M/(1-\tau)-2\kappa)+(\tilde M-16)(\tilde M/(1-\tau)-\kappa)+16\kappa]^2+\Or(N^{-2/3})}\leq e^{-c_3 m^2\tilde M^2/(1-\tau)^2} e^{-c_4 \tilde M^4/(1-\tau)^2}
\end{equation}
for all $\tilde M\geq 9$. From this estimate we then obtain
\begin{equation}
 \sum_{m=1}^\infty \Pb\Big(\sup_{v\in[m-1,m]} (S_v-S_m)\geq A-A_m\Big)\leq C e^{-c_4 (M-M_1)^4/(1-\tau)^2}.
\end{equation}
\end{proof}

\section{Convergence and first order correction of the covariance}\label{SectCovariance}
In this section we analyze the covariance of the point-to-point and of the stationary LPP at different ``times'', where time refers to the $(1,1)$-direction due to the well-known relation with KPZ growth models. In Theorem~\ref{thmCovStatFormula} we obtain an exact formula for the covariance for the stationary model if the two end-points are on the diagonal. In Section~\ref{SectUnivBeh} we study the behavior of the covariance when the two times becomes macroscopically close to each other. Specifically, we show that the first order correction is the same as the one of the stationary model, see Theorem~\ref{thmCovPP} and Theorem~\ref{thmCovPPGeneral}.

The difficulty is that the general point-to-point LPP process is not known and we can not assume that for instance its existence and tightness. Still, we are able to solve the problem with some careful thinking.

\subsection{Proof of Theorem~\ref{thmCovStatFormula}: the formula for the stationary case.}
We use the decomposition \eqref{eq5.5}. Due to the exponential tails bounds of Proposition~\ref{propAppBounds}, the variances also converges to the ones of their limiting random variables. What remains is to get an expression for the limit of $Y_N^{\rho}:={\cal L}_N^{\rm st,\rho}(0,\tau)-{\cal L}_N^{\rm st,\rho}(0,1)$. We have
\begin{equation}\label{eq5.54}
\begin{aligned}
\lim_{N\to\infty} Y_N^\rho&= \max_{u\geq 0}\left\{\tau^{1/3}[{\cal A}^{\rm st,hs}_{\delta\tau^{1/3}}(u/\tau^{2/3})-{\cal A}^{\rm st,hs}_{\delta\tau^{1/3}}(0)]+(1-\tau)^{1/3}[{\cal A}^{\rm pp}_{\delta(1-\tau)^{1/3}}(\tfrac{u}{(1-\tau)^{2/3}})-\tfrac{u^2}{(1-\tau)^{4/3}}]\right\}\\
 &= (1-\tau)^{1/3}\max_{v\geq 0}\left\{(\tfrac{\tau}{1-\tau})^{1/3}[{\cal A}^{\rm st,hs}_{\delta\tau^{1/3}}(v \tfrac{(1-\tau)^{2/3}}{\tau^{2/3}})-{\cal A}^{\rm st,hs}_{\delta\tau^{1/3}}(0)]+{\cal A}^{\rm pp}_{\delta(1-\tau)^{1/3}}(v)-v^2\right\},
 \end{aligned}
\end{equation}
where we changed the variables $v=u(1-\tau)^{-2/3}$. The exchange of the limit $N\to\infty$ and the maximum is justified because the probability that the maximum is reached at value $u\geq (1-\tau)^{2/3}\tilde M$ is $\Or(e^{-c\tilde M})$, see Proposition~\ref{PropCov} below, and the rescaled processes in the maximum converges weakly in the space of continuous functions with bounded support (see Theorem~\ref{thmWeakCvg} for the point-to-point case, while it is obvious for the stationary process).

${\cal A}^{\rm st,hs}_\delta(u)$ is the limit of sum of independent random variables (with a non-zero drift). Then, by Donsker's theorem, some simple computations lead to
\begin{equation}
{\cal A}^{\rm st,hs}_{\delta}(u+h)-{\cal A}^{\rm st,hs}_{\delta}(u) \stackrel{(d)}{=}\sqrt{2}B(h)+2 h \delta.
\end{equation}
Therefore
\begin{equation}
 (\tfrac{\tau}{1-\tau})^{1/3}[{\cal A}^{\rm st,hs}_{\delta\tau^{1/3}}(v \tfrac{(1-\tau)^{2/3}}{\tau^{2/3}})-{\cal A}^{\rm st,hs}_{\delta\tau^{1/3}}(0)] \stackrel{(d)}{=}\sqrt{2} B(v)+2v \delta (1-\tau)^{1/3},
\end{equation}
from which
\begin{equation}\label{eq5.11}
 \eqref{eq5.54} = (1-\tau)^{1/3}\max_{v\geq 0} \left\{\sqrt{2}B(v)+2v\delta_\tau+{\cal A}^{\rm pp}_{\delta_\tau}(v)-v^2\right\},
\end{equation}
where $\delta_\tau=\delta (1-\tau)^{1/3}$. On the other hand, taking $\tau=0$, we have the identity \eqref{eq5.12}. Consequently
\begin{equation}\label{eq5.10}
 \lim_{N\to\infty} Y_N^\rho \stackrel{(d)}{=}(1-\tau)^{1/3}{\cal A}^{\rm st,hs}_{\delta(1-\tau)^{1/3}}(0),
\end{equation}
which gives the claimed formula.

\subsection{Universal behavior for $\tau\to 1$}\label{SectUnivBeh}
In this section we will show that the third term in \eqref{eq5.5}, which is the first order correction, of the covariance for the point-to-point LPP is the same as the one for the stationary case with same parameter on the diagonal.

\begin{thm}\label{thmCovPP}
Let $\rho=\frac12+\delta 2^{-4/3} N^{-1/3}$, ${\cal L}_N^{\rm pp}$ as in \eqref{eq5.2} and ${\cal L}_N^{\rm st,\rho}$ as in \eqref{eq5.4}. Then, for any $0<\theta<1/3$, there exists a constant $C$ such that
\begin{equation}
\lim_{N\to\infty} \left|{\rm Var}({\cal L}_N^{\rm pp}(M_1,1)-{\cal L}_N^{\rm pp}(0,\tau))-{\rm Var}({\cal L}_N^{\rm st,\rho}(M_1,1)-{\cal L}_N^{\rm st,\rho}(0,\tau))\right| \leq C (1-\tau)^{1-\theta}
\end{equation}
as $\tau\to 1$.
\end{thm}
The variances are of order $(1-\tau)^{2/3}$ by Lemma~\ref{LemControlOrderForStationary} below. Thus the error term is about $(1-\tau)^{1/3}$ smaller than the actual value of the variance.

For the estimate, we will need to know the order of $\E(|{\cal L}_N^{\rm st,\rho}(M_1,1)-{\cal L}_N^{\rm st,\rho}(0,\tau)|)$ as $\tau\to 1$.
\begin{lem}\label{LemControlOrderForStationary}
As $\tau\to 1$, we have
\begin{equation}
\begin{aligned}
 \lim_{N\to\infty}\E(|{\cal L}_N^{\rm st,\rho}(M_1,1)-{\cal L}_N^{\rm st,\rho}(0,\tau)|)&=\Or((1-\tau)^{1/3}),\\
 \lim_{N\to\infty}\Var({\cal L}_N^{\rm st,\rho}(M_1,1)-{\cal L}_N^{\rm st,\rho}(0,\tau))&=\Or((1-\tau)^{2/3})
\end{aligned}
\end{equation}
for all $N$ large enough.
\end{lem}
\begin{proof}
We consider the upper bound
\begin{equation}
\E(|{\cal L}_N^{\rm st,\rho}(M_1,1)-{\cal L}_N^{\rm st,\rho}(0,\tau)|)\leq \E(|{\cal L}_N^{\rm st,\rho}(M_1,1)-{\cal L}_N^{\rm st,\rho}(0,1)|)+\E(|{\cal L}_N^{\rm st,\rho}(0,1)-{\cal L}_N^{\rm st,\rho}(0,\tau)|).
\end{equation}
From \eqref{eq5.10} we know that $\lim_{N\to\infty}\E(|{\cal L}_N^{\rm st,\rho}(0,1)-{\cal L}_N^{\rm st,\rho}(0,\tau)|)=\Or((1-\tau)^{1/3})$.

Furthermore, ${\cal L}_N^{\rm st,\rho}(M_1,1)-{\cal L}_N^{\rm st,\rho}(0,1)$ is a rescaled sum of independent random variables, which converges to a Brownian motion (with diffusion coefficient $2$) plus a finite drift. Since $M_1=\Or((1-\tau)^{2/3})$, we conclude that $\E(|{\cal L}_N^{\rm st,\rho}(M_1,1)-{\cal L}_N^{\rm st,\rho}(0,1)|)=\Or((1-\tau)^{1/3})$ as well.

For the variance, use the same decomposition and $\Var(A+B)\leq 2(\Var(A)+\Var(B))$ to get the claimed result.
\end{proof}

\subsection{Proof of Theorem~\ref{thmCovPP}}
For most of the proof the steps are valid for any value of $\tilde M_\tau$ and $\tilde M_1$. For this reason we state them for general values as they could be used to generalize Theorem~\ref{thmCovPP}.

Denote by $I(u)=(\tau N+u(2N)^{2/3},\tau N-u(2N)^{2/3})$. Then
\begin{equation}\label{eq5.20}
X_N={\cal L}_N^{\rm pp}(M_1,1)-{\cal L}_N^{\rm pp}(M_\tau,\tau)=\max_{u\geq 0} [{\cal L}_N^{\rm pp}(u,\tau)+{\cal L}_N^\rho (u,\tau;M_1,1)-{\cal L}_N^{\rm pp}(M_\tau,\tau))],
\end{equation}
where
\begin{equation}
 {\cal L}_N^\rho (u,\tau;M_1,1) = \frac{L(I(u),Q_1)-4(1-\tau) N}{2^{4/3}N^{1/3}}.
\end{equation}

We expect that the local behavior of the increments of ${\cal L}_N^{\rm pp}$ is of Brownian nature and thus as for the stationary case. But this not true on a large scale. Thus if the maximum in \eqref{eq5.20} is taken for too large values of $u$, then the claim would not be true. This is the reason why we define the random variables
\begin{equation}
\begin{aligned}
X_{N,M}&=\max_{0\leq u\leq M}[{\cal L}_N^{\rm pp}(u,\tau)+{\cal L}_N^\rho(u,\tau;M_1,1)-{\cal L}_N^{\rm pp}(M_\tau,\tau)],\\
X_{N,M^c}&=\max_{u>M} [{\cal L}_N^{\rm pp}(u,\tau)+{\cal L}_N^\rho(u,\tau;M_1,1)-{\cal L}_N^{\rm pp}(M_\tau,\tau)],
\end{aligned}
\end{equation}
and similarly for the stationary case with density $\rho$ (which has the same parameter as the point-to-point case),
\begin{equation}
\begin{aligned}
Y_N^\rho&=\max_{u\geq 0} [{\cal L}_N^{\rm st,\rho}(u,\tau)+{\cal L}_N^\rho(u,\tau;M_1,1)-{\cal L}_N^{\rm st,\rho}(M_\tau,\tau)],\\
Y_{N,M}^\rho&=\max_{0\leq u\leq M} [{\cal L}_N^{\rm st,\rho}(u,\tau)+{\cal L}_N^\rho(u,\tau;M_1,1)-{\cal L}_N^{\rm st,\rho}(M_\tau,\tau)],\\
Y_{N,M^c}^\rho&=\max_{u>M} [{\cal L}_N^{\rm st,\rho}(u,\tau)+{\cal L}_N^\rho(u,\tau;M_1,1)-{\cal L}_N^{\rm st,\rho}(M_\tau,\tau)].
\end{aligned}
\end{equation}

With these notations, we need to estimate
\begin{equation}
|\Var(X_N)-\Var(Y_N^\rho)|.
\end{equation}

\begin{remark}\label{rem5.1}
Below we will bound the increments ${\cal L}_N^{\rm pp}(u,\tau)-{\cal L}_N^{\rm pp}(M_\tau,\tau)$ by the increments of ${\cal L}_N^{\rm st,\rho}(u,\tau)-{\cal L}_N^{\rm st,\rho}(M_\tau,\tau)$ also for a different density, $\rho_-$ instead of $\rho$. In that case, when we write
\begin{equation}
 Y_N^{\rho_-}=\max_{u\geq 0} [{\cal L}_N^{\rm st,\rho_-}(u,\tau)+{\cal L}_N^\rho(u,\tau;M_1,1)-{\cal L}_N^{\rm st,\rho_-}(M_\tau,\tau)],
\end{equation}
\emph{the parameter on the diagonal for ${\cal L}_N^\rho(u,\tau;M_1,1)$ remains $\rho$}, as we replace only the increments until time $\tau N$.
\end{remark}

\paragraph{Step 1: Localization.}
The first step is to verify that the error term in the variance coming from localizing is small.
\begin{prop}\label{PropCov}
For all $\tilde M>0$, set $M=(1-\tau)^{2/3}\tilde M$. Then we have
\begin{equation}
\begin{aligned}
\Var(X_N)&=\Var(X_{N,M})+\Or(e^{-c \tilde M}),\\
\Var(Y_N^\rho)&=\Var(Y_{N,M}^\rho)+\Or(e^{-c  \tilde M})
\end{aligned}
\end{equation}
and
\begin{equation}
\begin{aligned}
\E(X_N)&=\E(X_{N,M})+\Or(e^{-c \tilde M}),\\
\E(Y_N^\rho)&=\E(Y_{N,M}^\rho)+\Or(e^{-c  \tilde M})
\end{aligned}
\end{equation}
uniformly in $N$ (and similarly for any other finite moments).
\end{prop}
\begin{proof}
By integration by parts one gets
\begin{equation}\label{eq5.13}
\E(X_N)=-\int_{-\infty}^0 ds \Pb(X_N\leq s)+\int_0^\infty ds \Pb(X_N>s)
\end{equation}
and
\begin{equation}\label{eq5.14}
\E(X_N^2)=-2\int_{-\infty}^0 ds \,s \Pb(X_N\leq s)+2\int_0^\infty ds\, s \Pb(X_N>s).
\end{equation}
Thus the variance is a linear combination of term of the form
\begin{equation}
 \int_{-\infty}^0 ds \,s^k\Pb(X_N\leq s),\quad \int_0^\infty ds \,s^k \Pb(X_N>s)
\end{equation}
with $k\in\{0,1\}$ only.

Using Lemma~\ref{lemCov} below, the proposition follows.
\end{proof}

\begin{lem}\label{lemCov} Let $M=(1-\tau)^{2/3}\tilde M$. We have
\begin{equation}
\begin{aligned}
\int_{-\infty}^0 ds\, s^k \Pb(X_N\leq s)&=\int_{-\infty}^0 ds\, s^k \Pb(X_{N,M}\leq s)-R_{\tilde M},\\
\int_0^\infty ds\, s^k \Pb(X_N> s)&=\int_0^\infty ds\, s^k \Pb(X_{N,M}> s)+\tilde R_{\tilde M}
\end{aligned}
\end{equation}
with
\begin{equation}
|R_{\tilde M}|\leq C e^{-c \tilde M},\quad |\tilde R_{\tilde M}|\leq C e^{-c \tilde M}.
\end{equation}
\end{lem}
\begin{proof}
Since $X_N=\max\{X_{N,M},X_{N,M^c}\}$ and $Y_N^\rho=\max\{Y_{N,M}^\rho,Y_{N,M^c}^\rho\}$, we have
\begin{equation}
\begin{aligned}
\Pb(X_N\leq s)&=\Pb(X_{N,M}\leq s)-\Pb(X_{N,M}\leq s,X_{N,M^c}>s),\\
\Pb(X_N> s)&=\Pb(X_{N,M}> s)+\Pb(X_{N,M}\leq s,X_{N,M^c}>s),
\end{aligned}
\end{equation}
from which
\begin{equation}
R_M=\int_{-\infty}^0 ds\, s^k \Pb(X_{N,M}\leq s,X_{N,M^c}>s).
\end{equation}
$R_M$ is bounded as
\begin{equation}
\begin{aligned}
|R_{\tilde M}|&= \int_{-\infty}^{-\tilde M} ds\, |s|^k \Pb(X_{N,M}\leq s,X_{N,M^c}>s)+\int_{-\tilde M}^0 ds\, |s|^k \Pb(X_{N,M}\leq s,X_{N,M^c}>s)\\
&\leq \int_{-\infty}^{-\tilde M} ds\, |s|^k \Pb(X_{N,M}\leq s)+\tilde M^k \Pb(X_{N,M}<X_{N,M^c}).
\end{aligned}
\end{equation}
As $X_{N,M}\geq {\cal L}_N^{\rm pp}(I(0))+{\cal L}_N^\rho(I(0),Q_1)-{\cal L}_N^{\rm pp}(I(M_\tau))$ and all the scaled random variables have at least exponential upper and lower tails, see Proposition~\ref{propAppBounds}, then also $X_{N,M}$ has exponential tails as well. Thus the first term is $\Or(\tilde M^k e^{-c \tilde M})$. For the second term, we use the bound on the localization of Proposition~\ref{propLocalization}, namely
\begin{equation}
\Pb(X_{N,M}<X_{N,M^c})=\Pb(\pi^{\rm pp}(Q_1)\not\prec I(M))\leq C e^{-c M^3/(1-\tau)^2}=C e^{-c \tilde M^3}.
\end{equation}
All in all, there exist constants $C,c>0$ such that
\begin{equation}
|R_{\tilde M}|\leq C e^{-c \tilde M}.
\end{equation}
Similarly, set
\begin{equation}
\tilde R_{\tilde M}=\int_0^\infty ds\, s^k \Pb(X_{N,M}\leq s,X_{N,M^c}>s).
\end{equation}
The bound we get on $\tilde R_{\tilde M}$ is
\begin{equation}
\begin{aligned}
|\tilde R_{\tilde M}|&= \int_{\tilde M}^\infty ds\, s^k \Pb(X_{N,M}\leq s,X_{N,M^c}>s)+\int_0^{\tilde M} ds\, s^k \Pb(X_{N,M}\leq s,X_{N,M^c}>s)\\
&\leq \int_{\tilde M}^\infty ds\, s^k \Pb(X_{N,M^c}> s)+\tilde M^k \Pb(X_{N,M}<X_{N,M^c})\\
&\leq \int_{\tilde M}^\infty ds\, s^k \Pb(X_{N}> s)+\tilde M^k \Pb(X_{N,M}<X_{N,M^c})
\end{aligned}
\end{equation}
As $X_N$ has exponentially decaying upper tail, we have $|\tilde R_{\tilde M}|\leq C e^{-c \tilde M}$ as well.
\end{proof}

\paragraph{Step 2: Threshold.}
Next we get an estimate on the variance of $X_{N,M}$ obtained by putting a threshold. This will be useful as the estimates on the distribution function of $X_{N,M}$ in terms of the ones of $Y^{\rho}_{N,M}$ will contain some errors which do not depend on its value, see e.g.~\eqref{eq5.42}. To avoid infinities when integrating over $\R_\pm$ we do first a cut-off with a threshold.
\begin{lem}\label{LemCutoff}
For all $\tilde M>0$, set $M=(1-\tau)^{2/3}\tilde M$. Then we have
\begin{equation}
\Var(X_{N,M})=\Var(X_{N,M} \Id_{|X_{N,M}|\leq \tilde M})+\Or(e^{-c \tilde M}),
\end{equation}
and the same holds for $Y^\rho_{N,M}$.
\end{lem}
\begin{proof}
For computing the variance of $X_{N,M}$ with threshold, we do integration by parts and we get
\begin{equation}\label{eq5.26}
\begin{aligned}
\E(X_{N,M}\Id_{|X_{N,M}|\leq \tilde M})=&-\int_{-\tilde M}^0 ds \Pb(X_{N,M}\leq s)+\int_0^{\tilde M} ds \Pb(X_{N,M}>s)-\tilde M \Pb(|X_{N,M}|>\tilde M),\\
\E(X_{N,M}^2\Id_{|X_{N,M}|\leq \tilde M})=&-2\int_{-\tilde M}^0 ds \, s \Pb(X_{N,M}\leq s)+2\int_0^{\tilde M} ds \, s \Pb(X_{N,M}>s)\\
&+\tilde M^2 (\Pb(X_{N,M}\leq -\tilde M)-\Pb(X_{N,M}>\tilde M).
\end{aligned}
\end{equation}
For the variance of $X_N$ we use \eqref{eq5.13} and \eqref{eq5.14}. Then, since $X_{N,M}$ has exponential decay in the upper and lower tail of its distribution, as we used in the proof of Lemma~\ref{lemCov}, the result is immediate.
\end{proof}

\paragraph{Step 3: Comparison with stationarity.}
Next we apply the comparison with stationarity, which are Corollary~\ref{CorCrossingA} and Proposition~\ref{PropCrossing}. Let $\rho_+=\rho=\tfrac12+\delta 2^{-4/3} N^{-1/3}$ and $\rho_-=\tfrac12+(\delta-\kappa) 2^{-4/3}N^{-1/3}$. Assume that $\kappa-\delta>6M$. Then for all $0\leq u_1<u_2\leq M$
\begin{equation}\label{eq5.24}
L^{\rho_-}(I(u_2))-L^{\rho_-}(I(u_1))\leq L^{\rm pp}(I(u_2))-L^{\rm pp}(I(u_1))\leq L^{\rho}(I(u_2))-L^{\rho}(I(u_1))
\end{equation}
on a set $\Omega_{\rm cross}$ with $\Pb(\Omega_{\rm cross})\geq 1-C e^{-c (\kappa-\delta)^3}$ for all $N$ large enough. In the end we are going to take $\kappa-\delta= C/(1-\tau)^{\theta/2}$ for some small $\theta>0$. Since $M=(1-\tau)^{2/3}\tilde M$, the condition $\kappa-\delta>6M$ will be satisfied for all $\tau$ close enough to $1$.

Next we decompose the random variable on the set $\Omega_{\rm cross}$ and $\Omega_{\rm cross}^c$,
\begin{equation}
 X_{N,M}=X_{N,M}\Id_{\Omega_{\rm cross}}+X_{N,M}\Id_{\Omega_{\rm cross}^c}
\end{equation}
and similarly for $Y^\rho_{N,M}$.
Also, we can write
\begin{equation}
X_{N,M}=\max\{X_{N,M_\tau},X_{N,(M_\tau,M]}\}
\end{equation}
where
\begin{equation}
\begin{aligned}
X_{N,M_\tau}&=\max_{0\leq u\leq M_\tau} [{\cal L}_N^{\rm pp}(u,\tau)+{\cal L}_N^\rho(u,\tau;M_1,1)-{\cal L}_N^{\rm pp}(M_\tau,\tau)],\\
X_{N,(M_\tau,M]}&=\max_{M_\tau<u\leq M} [{\cal L}_N^{\rm pp}(u,\tau)+{\cal L}_N^\rho(u,\tau;M_1,1)-{\cal L}_N^{\rm pp}(M_1,\tau)].
\end{aligned}
\end{equation}
Similarly define $Y_{N,M_\tau}^\rho$ and $Y_{N,(M_\tau,M]}^\rho$ for the stationary case with density $\rho$.

The inequalities \eqref{eq5.24} give, on the event ${\Omega_{\rm cross}}$,
\begin{equation}\label{eq5.36}
\begin{aligned}
Y_{N,M_\tau}^\rho&\leq X_{N,M_\tau}\leq Y_{N,M_\tau}^{\rho_-},\\
Y_{N,(M_\tau,M]}^{\rho_-}&\leq X_{N,(M_\tau,M]}\leq Y_{N,(M_\tau,M]}^{\rho}.
\end{aligned}
\end{equation}
Thus defining
\begin{equation}
Y_{N,M}^-=\max\{Y_{N,M_\tau}^{\rho},Y_{N,(M_\tau,M]}^{\rho_-}\}\quad \textrm{ and }\quad Y_{N,M}^+=\max\{Y_{N,M_\tau}^{\rho_-},Y_{N,(M_\tau,M]}^{\rho}\},
\end{equation}
we get
\begin{equation}
Y_{N,M}^-\Id_{\Omega_{\rm cross}} \leq X_{N,M}\Id_{\Omega_{\rm cross}} \leq Y_{N,M}^+ \Id_{\Omega_{\rm cross}}.
\end{equation}
From this we get
\begin{equation}\label{eq5.42}
\begin{aligned}
\Pb(Y_{N,M}^+>s)-\Pb(\Omega_{\rm cross}^c)&\leq \Pb(X_{N,M}>s)\leq \Pb(Y_{N,M}^->s)+\Pb(\Omega_{\rm cross}^c),\\
\Pb(Y_{N,M}^-\leq s)-\Pb(\Omega_{\rm cross}^c)&\leq \Pb(X_{N,M}\leq s) \leq \Pb(Y_{N,M}^+\leq s)+\Pb(\Omega_{\rm cross}^c).
\end{aligned}
\end{equation}
Note that the same inequalities holds if we replace $X_{N,M}$ with $Y^\rho_{N,M}$.

\paragraph{Step 4: Special case $M_\tau=0$.}
The expressions of $Y^-_{N,M}$ and $Y^+_{N,M}$ depends on the increments of both the stationary LPP with density $\rho$ and $\rho_-$. Thus, to get the best estimates on the correction term, we would need to use information on the coupled stationary processes, which are known for the full-space setting only~\cite{FS20}. There is however a case which can be analyzed without this further input, namely the $M_\tau=0$ case, which we now analyze.

In this case,
\begin{equation}\label{eq5.50}
\begin{aligned}
Y^-_{N,M}&=Y^{\rho_-}_{N,M}=\max_{0\leq u\leq M} [{\cal L}_N^{\rm st,\rho_-}(u,\tau)+{\cal L}_N^\rho(u,\tau;M_1,1)-{\cal L}_N^{\rm st,\rho_-}(0,\tau)],\\
Y^+_{N,M}&=Y^\rho_{N,M}=\max_{0\leq u\leq M} [{\cal L}_N^{\rm st,\rho}(u,\tau)+{\cal L}_N^\rho(u,\tau;M_1,1)-{\cal L}_N^{\rm st,\rho}(0,\tau)].
\end{aligned}
\end{equation}
We have
\begin{equation}
\begin{aligned}
 {\cal L}_N^{\rm st,\rho_-}(u,\tau)-{\cal L}_N^{\rm st,\rho_-}(0,\tau)=\frac{{L}^{\rm st,\rho_-}(I(u))-{L}^{\rm st,\rho_-}(I(0))}{2^{4/3}N^{1/3}}
 =\frac{1}{2^{4/3}N^{1/3}}\sum_{i=1}^{u(2N)^{2/3}}(\tilde X_i-\tilde Y_i),
\end{aligned}
\end{equation}
where $\tilde X_i\sim\exp(1-\rho_-)$ and $\tilde Y_i\sim\exp(\rho_-)$ are independent random variables. Similarly,
\begin{equation}\label{eq5.51}
\begin{aligned}
 {\cal L}_N^{\rm st,\rho}(u,\tau)-{\cal L}_N^{\rm st,\rho}(0,\tau)=\frac{1}{2^{4/3}N^{1/3}}\sum_{i=1}^{u(2N)^{2/3}}(X_i- Y_i),
\end{aligned}
\end{equation}
where $X_i\sim\exp(1-\rho)$ and $Y_i\sim\exp(\rho)$ are independent random variables.

Although for $M_\tau=0$ it is not needed, we actually also know that, given the coupling on the axis we considered, that $\tilde X_i\leq X_i$ and $\tilde Y_i\geq Y_i$. This tells us that $\tilde X_i-\tilde Y_i\leq X_i-Y_i$. Let $U_i\sim \textrm{Unif}((0,1])$ and $V_i\sim \textrm{Unif}((0,1])$ independent random variables. Then define
\begin{equation}
\begin{aligned}
\tilde X_i&:=-\frac1{1-\rho_-}\ln(U_i), \quad \hat X_i:=-\frac1{1-\rho}\ln(U_i),\\
\tilde Y_i&:=-\frac1{\rho_-}\ln(V_i), \quad \hat Y_i:=-\frac1{\rho}\ln(V_i).
\end{aligned}
\end{equation}
So we can decompose $\tilde X_i$ and $\tilde Y_i$ as follows:
\begin{equation}
\tilde X_i=\hat X_i-P_i,\quad \tilde Y_i=\hat Y_i+Q_i,
\end{equation}
where $\hat X_i\stackrel{(d)}{=}X_i$, $\hat Y_i\stackrel{(d)}{=}Y_i$ and $P_i\geq 0$, $Q_i\geq 0$.

Of course, $\hat X_i$ and $P_i$ are highly correlated, as well as $\hat Y_i$ and $Q_i$ are, but $P_i$ and $Q_i$ are independent. The law of $P_i$ and $Q_i$ are given by
\begin{equation}
\begin{aligned}
 \Pb(P_i\geq s)&=\Pb(\hat X_i-\tilde X_i\geq s) = \exp\left(-s \frac{(1-\rho)(1-\rho_-)}{\rho-\rho_-}\right),\\
 \Pb(Q_i\geq s)&=\Pb(\tilde Y_i-\hat Y_i\geq s) = \exp\left(-s \frac{\rho\rho_-}{\rho-\rho_-}\right).
\end{aligned}
\end{equation}

This means that
\begin{equation}\label{eq5.41}
{\cal L}_N^{\rm st,\rho_-}(u,\tau)-{\cal L}_N^{\rm st,\rho_-}(0,\tau)\stackrel{(d)}{=} {\cal L}_N^{\rm st,\rho}(u,\tau)-{\cal L}_N^{\rm st,\rho}(0,\tau) -
\frac{1}{2^{4/3}N^{1/3}}\sum_{i=1}^{ u(2N)^{2/3}}(P_i+Q_i)
\end{equation}
because ${\cal L}_N^{\rm st,\rho}(u,\tau)-{\cal L}_N^{\rm st,\rho}(0,\tau)\stackrel{(d)}{=}\frac{1}{2^{4/3}N^{1/3}}\sum_{i=1}^{u(2N)^{2/3}}(\hat X_i- \hat Y_i)$.
It is important to keep in mind that the two terms in the r.h.s.\ of \eqref{eq5.41} are not independent. However it is not a problem, because the latter goes to $0$ as $N\to\infty$.

Define
\begin{equation}
R(u)=\frac{1}{2^{4/3}N^{1/3}}\sum_{i=1}^{ u(2N)^{2/3}}(P_i+Q_i).
\end{equation}
Then
\begin{equation}
\E(R(u))=2 u \kappa + \Or(u \kappa^3 N^{-2/3})
\end{equation}
and
\begin{equation}
\Var(R(u))=u 2^{1/3} \kappa^2 N^{-2/3}+\Or(u \kappa^4 N^{-4/3}).
\end{equation}
From Corollary~\ref{CorA3} we have
\begin{equation}
\Pb(R(u)\geq 4 M\kappa)\leq C e^{-c M N^{2/3}}=C e^{-c \tilde M (1-\tau)^{2/3} N^{2/3}}.
\end{equation}
Define the event $\Omega_R=\{\sup_{0\leq u\leq M} R(u)< 4M \kappa\}$. Then, by union bound we get
\begin{equation}
\Pb(\Omega_R^c)\leq C M (2N)^{2/3} e^{-c M N^{2/3}}\leq \tilde C e^{-\tilde c \tilde M (1-\tau)^{2/3}N^{2/3}}
\end{equation}
for some new constant $\tilde C,\tilde c>0$.

Let $\e:=4M \kappa\sim (1-\tau)^{2/3}$. Decomposing on $\Id_{\Omega_R}$ and on $\Id_{\Omega_R^c}$ we get, uniformly for $u\in [0,M]$, that
\begin{equation}
\begin{aligned}
\Pb({\cal L}_N^{\rm st,\rho}(u,\tau)-{\cal L}_N^{\rm st,\rho}(0,\tau) >s)-\Pb(\Omega_R^c)&\leq \Pb({\cal L}_N^{\rm st,\rho_-}(u,\tau)-{\cal L}_N^{\rm st,\rho_-}(0,\tau)>s)\\
&\leq \Pb({\cal L}_N^{\rm st,\rho}(u,\tau)-{\cal L}_N^{\rm st,\rho}(0,\tau) >s+\e)+\Pb(\Omega_R^c),
\end{aligned}
\end{equation}
and
\begin{equation}
\begin{aligned}
\Pb({\cal L}_N^{\rm st,\rho}(u,\tau)-{\cal L}_N^{\rm st,\rho}(0,\tau) \leq s-\e)-\Pb(\Omega_R^c)&\leq \Pb({\cal L}_N^{\rm st,\rho_-}(u,\tau)-{\cal L}_N^{\rm st,\rho_-}(0,\tau)\leq s) \\
&\leq \Pb({\cal L}_N^{\rm st,\rho}(u,\tau)-{\cal L}_N^{\rm st,\rho}(0,\tau) \leq s)+\Pb(\Omega_R^c).
\end{aligned}
\end{equation}
As the process ${\cal L}_N^\rho(u,\tau;M_1,1)$ in \eqref{eq5.50} is independent of the stationary increments, we conclude that
\begin{equation}
\begin{aligned}
\Pb(Y^\rho_{N,M}>s)-\Pb(\Omega_R^c)&\leq \Pb(Y^{\rho_-}_{N,M}>s)\leq \Pb(Y^\rho_{N,M}>s+\e)+\Pb(\Omega_R^c),\\
\Pb(Y^\rho_{N,M}\leq s-\e)-\Pb(\Omega_R^c)&\leq \Pb(Y^{\rho_-}_{N,M}\leq s)\leq \Pb(Y^\rho_{N,M}\leq s)+\Pb(\Omega_R^c),
\end{aligned}
\end{equation}
and, together with \eqref{eq5.42}, we have
\begin{equation}
\begin{aligned}
&\Pb(Y^{\rho}_{N,M}>s)-\Pb(\Omega_{\rm cross}^c)\leq \Pb(X_{N,M}>s)\leq \Pb(Y_{N,M}^\rho>s+\e)+\Pb(\Omega_R^c)+\Pb(\Omega_{\rm cross}^c),\\
&\Pb(Y_{N,M}^\rho\leq s-\e)-\Pb(\Omega_R^c)-\Pb(\Omega_{\rm cross}^c)\leq \Pb(X_{N,M}\leq s) \leq \Pb(Y_{N,M}^{\rho}\leq s)+\Pb(\Omega_{\rm cross}^c).
\end{aligned}
\end{equation}

We now apply the inequalities in the different terms of \eqref{eq5.26}. We get
\begin{equation}
\begin{aligned}
\E(X_{N,M}\Id_{|X_{N,M}|\leq \tilde M}) &\leq \E(Y^\rho_{N,M}\Id_{|Y^\rho_{N,M}|\leq \tilde M+\e}) +\e+2\tilde M [\Pb(\Omega_R^c)+\Pb(\Omega_{\rm cross}^c)]+\Or(e^{-c \tilde M}),\\
\E(X_{N,M}\Id_{|X_{N,M}|\leq \tilde M}) &\geq \E(Y^\rho_{N,M}\Id_{|Y^\rho_{N,M}|\leq \tilde M}) -2\tilde M \Pb(\Omega_{\rm cross}^c)+\Or(e^{-c \tilde M})
\end{aligned}
\end{equation}
and
\begin{equation}
\begin{aligned}
\E(X_{N,M}^2\Id_{|X_{N,M}|\leq \tilde M})&\leq \E((Y^\rho_{N,M})^2\Id_{|Y^\rho_{N,M}|\leq \tilde M+\e})+\e^2-2\e \E(Y^\rho_{N,M}\Id_{0\leq Y^\rho_{N,M}\leq \tilde M+\e})+\Or(e^{-c \tilde M}),\\
\E(X_{N,M}^2\Id_{|X_{N,M}|\leq \tilde M})&\geq \E((Y^\rho_{N,M})^2\Id_{|Y^\rho_{N,M}|\leq \tilde M})-\e^2+2\e \E(Y^\rho_{N,M}\Id_{-\tilde M -\e \leq Y^\rho_{N,M}\leq 0})+\Or(e^{-c \tilde M}).
\end{aligned}
\end{equation}
The error terms $\Or(e^{-c \tilde M})$ comes from the boundary terms in the integration by parts of \eqref{eq5.26} and the fact that the distribution functions of the random variables we consider have at least exponential lower and upper tails.

These estimates together with Proposition~\ref{PropCov} and Lemma~\ref{LemCutoff} lead to
\begin{equation}
|\Var(X_{N})-\Var(Y^\rho_{N})|=\Or\left(\e\E(|Y^\rho_{N}|);\e^2; \tilde M \Pb(\Omega_R^c);\tilde M \Pb(\Omega_{\rm cross}^c)\right)
\end{equation}
with $\e = 4M \kappa=4\kappa (1-\tau)^{2/3}\tilde M$. Taking $\kappa=\tilde M = 1/(1-\tau)^{\theta/2}$, with $0<\theta<1/3$, one gets
\begin{equation}
|\Var(X_{N})-\Var(Y^\rho_{N})|\leq C (1-\tau)^{2/3-\theta}\E(|Y^\rho_{N}|)
\end{equation}
as $\tau\to 1$. Finally, by Lemma~\ref{LemControlOrderForStationary}, we know that $\E(|Y^\rho_{N}|)=\Or((1-\tau)^{1/3})$. Taking $N\to\infty$, the proof of Theorem~\ref{thmCovPP} is completed.

\subsection{Proof of Theorem~\ref{thmCovPPGeneral}}
Consider the general case $M_\tau>0$ and denote by $X_{N,M}^{(0)}=X_{N,M}\big\vert_{M_\tau=0}$. Then we have
\begin{equation}
X_{N,M} = X_{N,M}^{(0)} -Z_N^{\rm pp}\quad \textrm{with }Z_N^{\rm pp}={\cal L}_N^{\rm pp}(M_\tau,\tau)-{\cal L}_N^{\rm pp}(0,\tau)
\end{equation}
and similarly define $Z_N^\rho$. Then,
\begin{equation}\label{eq5.68}
\begin{aligned}
\Var(X_{N,M})&=\Var(X_{N,M}^{(0)})+\Var(Z_N^{\rm pp})-2 \Cov(X_{N,M}^{(0)},Z_N^{\rm pp}),\\
\Var(Y^\rho_{N,M})&=\Var(Y_{N,M}^{\rho,(0)})+\Var(Z_N^{\rho})-2 \Cov(Y_{N,M}^{\rho,(0)},Z_N^{\rho}).
\end{aligned}
\end{equation}
Decompose the increments at time $\tau$ as
\begin{equation}\label{eq5.69}
Z_N^{\rm pp}=Z_N^\rho- \Delta_N,\quad \Delta_N= {\cal L}_N^{\rm st,\rho}(M_\tau,\tau)-{\cal L}_N^{\rm st,\rho}(0,\tau)-({\cal L}_N^{\rm pp}(M_\tau,\tau)-{\cal L}_N^{\rm pp}(0,\tau)).
\end{equation}
Note that Corollary~\ref{CorCrossingA} gives $\Delta_N\geq 0$ and Proposition~\ref{PropCrossing}
\begin{equation}\label{eq5.70}
\Delta_N\leq {\cal L}_N^{\rm st,\rho}(M_\tau,\tau)- {\cal L}_N^{\rm st,\rho}(0,\tau)-({\cal L}_N^{\rm st,\rho_-}(M_\tau,\tau)- {\cal L}_N^{\rm st,\rho_-}(0,\tau))
\end{equation}
on $\Omega_{\rm cross}$. As in the proof of Theorem~\ref{thmCovPP}, the contribution to the variance of the event $\Omega_{\rm cross}^c$ are irrelevant when $\tau\to 1$. We do not redo the details since they are almost identical to the ones in the previous proof. Instead below we implicitly consider that \eqref{eq5.70} holds always.

By Donsker's theorem, as $N\to\infty$,
\begin{equation}\label{eq5.73}
\begin{aligned}
 {\cal L}_N^{\rm st,\rho_-}(M_\tau,\tau)- {\cal L}_N^{\rm st,\rho_-}(0,\tau) &\longrightarrow \sqrt{2}(1-\tau)^{1/3} {\cal B}(\tilde M_\tau)+2(\delta-\kappa) (1-\tau)^{2/3}\tilde M_\tau,\\
 {\cal L}_N^{\rm st,\rho}(M_\tau,\tau)-{\cal L}_N^{\rm st,\rho}(0,\tau)&\longrightarrow \sqrt{2} (1-\tau)^{1/3}\tilde {\cal B}(\tilde M_\tau)+2\delta (1-\tau)^{2/3}\tilde M_\tau,
\end{aligned}
\end{equation}
where $\cal B$ and $\tilde {\cal B}$ are (not independent) standard Brownian motions.

Taking the difference of the two expressions in \eqref{eq5.68}, using the decomposition~\eqref{eq5.69} and then Cauchy-Schwarz to estimate the covariance by variances, we get
\begin{equation}\label{eq5.75}
 \begin{aligned}
|\Var(X_{N,M})-\Var(Y^\rho_{N,M})|\leq& |\Var(X_{N,M}^{(0)})-\Var(Y_{N,M}^{\rho,(0)})| + \Var(\Delta_N)+2 |\Cov(Z_N^\rho,\Delta_N)|\\
&+2|\Cov(X_{N,M}^{(0)},\Delta_N)|+2|\Cov(X^{(0)}_{N,M}-Y_{N,M}^{(0),\rho},Z_N^\rho)|\\
\leq &|\Var(X_{N,M}^{(0)})-\Var(Y_{N,M}^{\rho,(0)})| + \Var(\Delta_N) + 2 \sqrt{\Var(Z_N^\rho)\Var(\Delta_N)}\\
&+2\sqrt{\Var(X_{N,M}^{(0)})\Var(\Delta_N)}+ 2 \sqrt{\Var(X^{(0)}_{N,M}-Y_{N,M}^{(0),\rho})\Var(Z_N^\rho)}.
\end{aligned}
\end{equation}
Recall that we take $\kappa=\tilde M = 1/(1-\tau)^{\theta/2}$ for some small $\theta>0$. Then the single terms are bounded as follows:
\begin{itemize}
\item[(a)] Proposition~\ref{PropCov} and Theorem~\ref{thmCovPP} give $|\Var(X_{N,M}^{(0)})-\Var(Y_{N,M}^{\rho,(0)})|=\Or((1-\tau)^{1-\theta})$.
\item[(b)] From \eqref{eq5.73} we have $\Var(Z_N^\rho)=\Or((1-\tau)^{2/3})$.
\item[(c)] Proposition~\ref{PropOrderedRV} with $B={\cal L}_N^{\rm st,\rho_-}(0,\tau)- {\cal L}_N^{\rm st,\rho_-}(M_\tau,\tau)$ and $A={\cal L}_N^{\rm st,\rho}(0,\tau)- {\cal L}_N^{\rm st,\rho}(M_\tau,\tau)$ gives $\Var(\Delta_N)\leq \E(\Delta_N^2)=\Or((1-\tau)^{4/5} \tilde M\sqrt{\kappa})=\Or((1-\tau)^{4/5-3\theta/4})$.
\item[(d)] From Proposition~\ref{PropCov} and Lemma~\ref{LemControlOrderForStationary} we have $\Var(Y_{N,M}^{\rho,(0)})=\Or((1-\tau)^{2/3})$ and together with (a) we get $|\Var(X_{N,M}^{(0)})|=\Or((1-\tau)^{2/3})$.
\item[(e)] For $Y_{N,M}^{(0),\rho}-X^{(0)}_{N,M}$, note that on $\Omega_{\rm cross}$, $0\leq Y_{N,M}^{(0),\rho}-X^{(0)}_{N,M}\leq Y_{N,M}^{(0),\rho}-Y^{(0),\rho_-}_{N,M}$. We apply Proposition~\ref{PropOrderedRV} with $A=Y^{(0),\rho_-}_{N,M}$ and $B=Y^{(0),\rho}_{N,M}$. As input we have
\begin{equation}\label{eq5.76}
\begin{aligned}
&\E(Y^{(0),\rho}_{N})-\E(Y^{(0),\rho_-}_{N})\\
&=\E({\cal L}_N^{\rm st,\rho}(M_1,1))+\E({\cal L}_N^{\rm st,\rho}(0,\tau))-\E({\cal L}_N^{\rm st,\rho_-}(M_1,1))-\E({\cal L}_N^{\rm st,\rho_-}(0,\tau)).
\end{aligned}
\end{equation}
Due to stationarity, the expected value of each term is explicit (see \eqref{eq2.41}) and, as $N\to\infty$ (see \eqref{eqMeanStatProc}) we get
\begin{equation}
\begin{aligned}
\eqref{eq5.76} &\to -(\delta-\kappa)^2(1-\tau)-4(\delta-\kappa)\tilde M_1 (1-\tau)^{2/3} + \delta^2(1-\tau)+4\delta\tilde M_1 (1-\tau)^{2/3}\\
&=-\delta(\delta-2\kappa)(1-\tau)+4\kappa\tilde M_1 (1-\tau)^{2/3}.
\end{aligned}
\end{equation}
By Proposition~\ref{PropCov}, $\E(B-A)=-\delta(\delta-2\kappa)(1-\tau)+4\kappa\tilde M_1 (1-\tau)^{2/3}+ \Or(e^{-c/(1-\tau)^{\theta/2}})$.
Furthermore, by Proposition~\ref{PropCov}, $\E(B^4)=\E((Y^{(0),\rho}_{N})^4)+\Or(e^{-c \tilde M})$.

As in Lemma~\ref{LemControlOrderForStationary}, we decompose $Y^{(0),\rho}_{N}=a+b$ with $a={\cal L}_N^{\rm st,\rho}(M_1,1)-{\cal L}_N^{\rm st,\rho}(0,1)$, $b={\cal L}_N^{\rm st,\rho}(0,1)-{\cal L}_N^{\rm st,\rho}(0,\tau)$. Then we use the bound $(a+b)^4\leq 10 (a^4+b^4)$.

As in \eqref{eq5.73}, $a\to \sqrt{2} (1-\tau)^{1/3} {\cal B}(\tilde M_1)+2\delta(1-\tau)^{2/3} \tilde M_1$ as $N\to\infty$, so that we have $\E(a^4)\to C_1 (1-\tau)^{4/3}(1+\Or((1-\tau)^{2/3}))$. Finally, by \eqref{eq5.10}, we know that $b\to (1-\tau)^{1/3} {\cal A}^{\rm st,hs}_{\delta(1-\tau)^{1/3}}(0)$ as $N\to\infty$, from which $\E(b^4) \to\Or((1-\tau)^{4/3})$.

Proposition~\ref{PropOrderedRV} leads to $\Var(X^{(0)}_{N,M}-Y_{N,M}^{(0),\rho})\leq \E((X^{(0)}_{N,M}-Y_{N,M}^{(0),\rho})^2) = \Or((1-\tau)^{4/5})$.
\end{itemize}

Inserting all the estimates in this list into \eqref{eq5.75} we finally get
\begin{equation}
 |\Var(X_{N,M})-\Var(Y^\rho_{N,M})|\leq \Or((1-\tau)^{11/15-3\theta/8}).
\end{equation}
Combining this with Proposition~\ref{PropCov} and renaming $3\theta/8=\eta$ we get the claimed result.

\appendix

\section{Random walk bounds}
\begin{lem}\label{LemmaBoundRandomWalks}
Let us consider $Z_k\sim X_k-Y_k$ where $X_k\sim \exp(1-\rho)$, $Y_k\sim \exp(\rho)$, $k\geq 1$ are all independent random variables.
Let $\rho=\frac12+\varkappa\, 2^{-4/3} N^{-1/3}$ and
\begin{equation}
W_N(L)=\frac{\sum_{k=1}^{L (2N)^{2/3}} Z_k- 4\, 2^{1/3} L \varkappa N^{1/3}}{2^{4/3} N^{1/3}}.
\end{equation}
Then for all $\xi>0$ with $\xi=o(N^{1/6})$,
\begin{equation}
 \Pb(W_N(L)\leq -\xi)\leq 2 e^{-\xi^2/(4L)}
\end{equation}
for all $N$ large enough. Similarly
\begin{equation}
 \Pb(W_N(L)\geq \xi)\leq 2 e^{-\xi^2/(4L)}
\end{equation}
and
\begin{equation}
 \Pb\Big(\sup_{u\in[0,L]}W_N(u)\geq \xi\Big)\leq e^{-\xi^2/(4L)}.
\end{equation}
\end{lem}
\begin{proof}
By standard exponential Chebyshev inequality we have
\begin{equation}
\Pb(W_N(L)\leq -\xi)\leq \inf_{\lambda>0}\frac{\E(e^{-\lambda W_N(L)})}{e^{\lambda \xi}}=\inf_{\lambda>0} \frac{\E(e^{-\lambda 2^{-4/3} N^{-1/3} Z_1})^{L(2N)^{2/3}}}{e^{\lambda (\xi-2L\varkappa)}}
\end{equation}
where we used independence of the $Z_k$.

As $\E(e^{-\mu Z_1})=\frac{1-\rho}{1-\rho+\mu}\frac{\rho}{\rho-\mu}$, with $\mu=\lambda2^{-4/3}N^{-1/3}$ and the choice $\lambda=\xi/(2L)$ (which is approximately the minimum) we get
\begin{equation}
\Pb(W_N(L)\leq -\xi)\leq e^{-\xi^2/4L} e^{N^{-2/3}\Or(\varkappa^3 \xi;\xi^4/L^3)}.
\end{equation}

Thus as soon as $\xi \varkappa^3/L=o(N^{2/3})$ and $\xi^4/L^3=o(N^{2/3})$, which holds for $\xi=o(N^{1/6})$ as $N\to\infty$, the exponential of the error term is is bounded by $2$ for all $N$ large enough. Similarly one gets the second bound.

For the last bound we recall that $\mathcal{M}_u=\sum_{k=1}^{u(2N)^{2/3}}Z_k$ is a submartingale, and so it is $\exp(t\mathcal{M}_u)$ for $t>0$. We can use Doob’s inequality for submartingales,
\begin{equation}\label{eq:eqA1}
\begin{aligned}
\Pb\Big(\sup_{u\in[0,L]}W_N(u)\geq \xi\Big)&=\Pb\Big(\sup_{u\in[0,L]}\mathcal M_u\geq 2^{4/3}N^{1/3}\xi+4 2^{1/3}L\varkappa N^{1/3}\Big)\\
 &\leq \inf_{\lambda\geq0}\frac{\E\big(e^{\lambda 2^{-4/3}N^{-1/3}\mathcal{M}_L}\big)}{e^{t(\xi+2L\varkappa)}}=\inf_{\lambda\geq0}\frac{\E \big(e^{\lambda 2^{-4/3}N^{-1/3}Z_1}\big)^{L(2N)^{2/3}}}{e^{\lambda(\xi+2L\varkappa)}}.
 \end{aligned}
\end{equation}
The same computations as above give that the term in \eqref{eq:eqA1} is bounded by $e^{-\xi^2/4 L}$, with the only difference the sign of $\mu$ in $\E(e^{\mu Z_1})=\frac{1-\rho}{1-\rho-\mu}\frac{\rho}{\rho+\mu}$.
\end{proof}

\begin{lem}\label{lemA2}
Let $\lambda_1=\varkappa_1^{-1}\, 2^{-2/3} N^{1/3}$ and $\lambda_2=\varkappa_2^{-1}\, 2^{-2/3} N^{1/3}$. Let $P_i$ and $Q_i$ be independent random variables with
\begin{equation}
P_i\sim \exp(\lambda_1),\quad Q_i\sim \exp(\lambda_2).
\end{equation}
Then
\begin{equation}
\Pb\bigg(\frac{1}{2^{4/3}N^{1/3}}\sum_{i=1}^{u(2N)^{2/3}}(P_i+Q_i) \geq u (\varkappa_1+\varkappa_2)+s N^{-1/3}\bigg)\leq C \exp\left(-\frac{s^2}{2^{1/3}u (\varkappa_1^2+\varkappa_2^2)}\right)
\end{equation}
provided $s=o(N^{1/3} u\max\{\varkappa_1,\varkappa_2\})$.
\end{lem}
\begin{proof}
The proof follows the same pattern as the one of Lemma~\ref{LemmaBoundRandomWalks} and thus we refrain writing the details.
\end{proof}

\begin{cor}\label{CorA3}
Let $\rho=\frac12+\delta 2^{-4/3} N^{-1/3}$ and $\rho_-=\frac12+(\delta-\kappa) 2^{-4/3} N^{-1/3}$. Let $P_i$ and $Q_i$ be independent random variables with
\begin{equation}
 P_i\sim \exp\left(\tfrac{(1-\rho)(1-\rho_-)}{\rho-\rho_-}\right),\quad Q_i\sim \exp\left(\tfrac{\rho\rho_-}{\rho-\rho_-}\right).
\end{equation}
Then,
\begin{equation}
\Pb\bigg(\frac{1}{2^{4/3}N^{1/3}}\sum_{i=1}^{(u(2N)^{2/3}}(P_i+Q_i)\geq 2u\kappa+s N^{-1/3}\bigg) \leq  C \exp\left(-\frac{s^2}{2^{4/3}u \kappa^2}\right)
\end{equation}
provided $s=o(N^{1/3} u \kappa)$.
\end{cor}
\begin{proof}
We have
\begin{equation}
\frac{(1-\rho)(1-\rho_-)}{\rho-\rho_-} = \kappa^{-1} 2^{-2/3}n^{-1/3}+\Or(1),\quad \frac{\rho\rho_-}{\rho-\rho_-}=\kappa^{-1} 2^{-2/3} n^{-1/3}+\Or(1).
\end{equation}
Then the result follows from Lemma~\ref{lemA2} with $\varkappa_1=\varkappa_2=\kappa$.
\end{proof}

\section{A bound on ordered random variables}
\begin{prop}\label{PropOrderedRV}
Let $A,B$ be random variables satisfying
\begin{equation}
B\geq A,\quad \E(B^4)=(1-\tau)^{4/3} C_1,\quad \E(B-A)=(1-\tau)^{2/3}K.
\end{equation}
Then, for any $R>0$,
\begin{equation}
\E((B-A)^2)\leq R^2 (1-\tau)^{4/3} + (1-\tau)^{2/3} \sqrt{C_1 K / R}.
\end{equation}
Consequently, for $R=(1-\tau)^{-4/15}$,
\begin{equation}
\E((B-A)^2)\leq (1-\tau)^{4/5}(1+\sqrt{C_1 K}).
\end{equation}
\end{prop}
\begin{proof}
We have
\begin{equation}
\begin{aligned}
 \E((B-A)^2)&=\E\left((B-A)^2 \Id_{B-A< R (1-\tau)^{2/3}}\right)+\E\left((B-A)^2 \Id_{B-A\geq R (1-\tau)^{2/3}}\right)\\
&\leq R^2 (1-\tau)^{4/3}+4\E\left(B^2 \Id_{B-A\geq R (1-\tau)^{2/3}}\right)\\
&\leq R^2 (1-\tau)^{4/3}+4\sqrt{\E(B^4)}\sqrt{\Pb(B-A\geq R(1-\tau)^{2/3})},
\end{aligned}
\end{equation}
where we used Cauchy-Schwarz inequality in the last step. Inserting the estimate from Markov inequality
\begin{equation}
 \Pb(B-A\geq R(1-\tau)^{2/3}) \leq \frac{\E(B-A)}{R(1-\tau)^{2/3}}= \frac{K}{R}
\end{equation}
we obtain the claimed result.
\end{proof}

\section{Bounds for Laguerre $\beta$-ensembles}\label{appLaguerreBeta}
Consider the Laguerre $\beta$-ensemble with parameters $m+1>n\geq 1$ whose density on ordered $\lambda_1\leq \lambda_2\leq \ldots\leq \lambda_n$ is proportional to
\begin{equation}
\prod_{1\leq i<j\leq n}|\lambda_j-\lambda_i|^\beta\prod_{i=1}^n \lambda_i^{\frac{\beta}{2}(m-n+1)-1} e^{-\frac{\beta}{2}\lambda_i}.
\end{equation}
In Theorem~2 of~\cite{LR10} it is shown\footnote{There is a minor typo in the Laguerre weight function~\cite{LR10}, compare with the cited papers in there.} that for all $\beta\geq 1$ and $0<\e\leq 1$ and $m\geq n$, there are constants $C,c>0$ such that
\begin{equation}
\begin{aligned}
\Pb(\lambda_n\geq (\sqrt{m}+\sqrt{n})^2(1+\e))&\leq C e^{-c\beta \e^{3/2}(m n)^{1/2}\min\{\e^{-1/2},(m/n)^{1/4}\}},\\
\Pb(\lambda_n\leq (\sqrt{m}+\sqrt{n})^2(1-\e))&\leq C^\beta e^{-c\beta \e^3 m n\min\{\e^{-1},(m/n)^{1/2}\}}.
\end{aligned}
\end{equation}
In particular, for $m=n+\Or(1)$, for all $\e=o(1)$ we have
\begin{equation}\label{eqC.3}
\begin{aligned}
\Pb(\lambda_n\geq (\sqrt{m}+\sqrt{n})^2(1+\e))&\leq C e^{-c\beta \e^{3/2}n},\\
\Pb(\lambda_n\leq (\sqrt{m}+\sqrt{n})^2(1-\e))&\leq C^\beta e^{-c\beta \e^3 n^2}.
\end{aligned}
\end{equation}

The relations with LPP we are interested in are the following ones:
\begin{itemize}
\item[(a)] for $\beta=2$ and $m=n$,
\begin{equation}
 \lambda_n\stackrel{(d)}{=}L^\square(n,n),
\end{equation}
where $L^\square$ denotes the LPP in full-space with $\omega_{i,j}\sim \exp(1)$ for all $i,j$ (see Proposition~1.4 of~\cite{Jo00b}).
Then \eqref{eqC.3} gives, for $0<s=o(N^{1/3})$,
\begin{equation}\label{eqPPfullSpacebounds}
\begin{aligned}
\Pb(L^\square(N,N)\geq 4N+s N^{1/3})&\leq C e^{-c s^{3/2}},\\
\Pb(L^\square(N,N)\leq 4N-s N^{1/3})&\leq C e^{-c s^3},
\end{aligned}
\end{equation}
for some other constants $C,c>0$.
\item[(b)] for $\beta=4$, $m=n-1/2$ and $n=N/2$ with even $N$, by Corollary~1.4 of~\cite{Bai02} (after some minor change of variables) we have
\begin{equation}\label{eqC.6}
 2\lambda_n\stackrel{(d)}{=}L^{{\rm pp};\rho=\infty}(N,N)\stackrel{(d)}{=}L^{{\rm pp};\rho=1}(N-1,N-1)
\end{equation}
where with $L^{{\rm pp};\rho}$ we mean the LPP in half-space with weight on the diagonal distributed as $\exp(\rho)$.
Then \eqref{eqC.3} gives\footnote{For the equality in law \eqref{eqC.6} one needs $N$ to be even. However the asymptotics for the LPP is clearly the same for even or odd $N$, so that the bounds holds true also for all $N$.}, for $0<s=o(N^{1/3})$,
\begin{equation}\label{eqPPhalfSpacebounds}
\begin{aligned}
\Pb(L^{{\rm pp};\rho=\infty}(N,N)\geq 4N+s N^{1/3})&\leq C e^{-c s^{3/2}},\\
\Pb(L^{{\rm pp};\rho=\infty}(N,N)\leq 4N-s N^{1/3})&\leq C e^{-c s^3},
\end{aligned}
\end{equation}
for some other constants $C,c>0$.
\end{itemize}

\section{Rough bounds on the upper and lower tails}\label{AppRoughBounds}
Denote by $L^{\rm st,\rho}$ stationary LPP with parameter $\rho$ and by $\exp(\rho)$, $L^{{\rm pp};\rho}$ the point-to-point LPP with $\exp(\rho)$ random variables on the diagonal.
\begin{prop}\label{propAppBounds}
Consider $n=M_1 (2N)^{2/3}$ and $\rho=\tfrac12+\delta 2^{-4/3} N^{-1/3}$, with $M_1$ and $\delta$ fixed. Then there exist constants $C,c>0$ such that for all $S\geq 0$,
\begin{equation}\label{eqB.2bis}
\begin{aligned}
\Pb(L^{\rm st,\rho}(N+n,N-n) &\geq 4N+S 2^{4/3} N^{1/3})\leq C e^{-c S},\\
\Pb(L^{{\rm pp};\rho}(N+n,N-n) &\geq 4N+S 2^{4/3} N^{1/3})\leq C e^{-c S},
\end{aligned}
\end{equation}
and
\begin{equation}\label{eqB.3bis}
\begin{aligned}
\Pb(L^{\rm st,\rho}(N+n,N-n)&\leq 4N-S 2^{4/3}N^{1/3})\leq C e^{-c S^{3/2}},\\
\Pb(L^{{\rm pp};\rho}(N+n,N-n)&\leq 4N-S 2^{4/3}N^{1/3})\leq C e^{-c S^{3/2}},
\end{aligned}
\end{equation}
uniformly in $N$.
\end{prop}
\begin{proof}
As in all the rest of this paper, the different LPP are coupled by setting the same random variables in the bulk $\cal B$ and by ordering the exponential random variables on the boundaries $\cal R$. Then, for all $N$ large enough,
\begin{equation}
L^{\rm st,\rho}(N+n,N-n)\geq L^{{\rm pp};\rho}(N+n,N-n)\geq L^{{\rm pp};1}(N+n,N-n)
\end{equation}
for all $n\geq 0$. Therefore a bound on the lower tail of $L^{{\rm pp};1}(N+n,N-n)$ will be also a bound on the lower tail of $L^{\rm st,\rho}(N+n,N-n)$ and $L^{{\rm pp};\rho}(N+n,N-n)$. Also, a bound on the upper tail of $L^{\rm st,\rho}(N+n,N-n)$ will also be a bound on the upper tail of $L^{{\rm pp};\rho}(N+n,N-n)$. These are given in Lemma~\ref{lemAppBounds} below.
\end{proof}

\begin{lem}\label{lemAppBounds}
Consider $n=M_1 (2N)^{2/3}$ and $\rho=\tfrac12+\delta 2^{-4/3} N^{-1/3}$, with $M_1$ and $\delta$ fixed. Then there exist constants $C,c>0$ such that for all $S\geq 0$,
\begin{equation}\label{eqB.2}
\Pb(L^{\rm st,\rho}(N+n,N-n) \geq 4N+S 2^{4/3} N^{1/3})\leq C e^{-c S}
\end{equation}
and
\begin{equation}\label{eqB.3}
\Pb(L^{{\rm pp};1}(N+n,N-n)\leq 4N-S 2^{4/3}N^{1/3})\leq C e^{-c S^{3/2}}
\end{equation}
uniformly in $N$.
\end{lem}
\begin{proof}
There are at least two ways to obtain the bound \eqref{eqB.2}. The first one consists in using the inequalities: (a) for $\rho\geq \frac12$, $L^{\rm st,\rho}(N+n,N-n)\leq L^{\beta=1/2,\alpha=1-\rho}(N+n,N-n)$ and (b) for $\rho\leq \frac12$, $L^{\rm st,\rho}(N+n,N-n)\leq L^{\beta=\rho,\alpha=1/2}(N+n,N-n)$, where $L^{\beta,\alpha}$ is the half-space LPP with weights
\begin{equation}\label{eqWeightsAlphaBeta}
\left\{
\begin{array}{ll}
\omega^\rho_{0,0} \sim \exp(\alpha+\beta), &\\
\omega^\rho_{i,i} \sim \exp(\frac12+\alpha), &\text{ for }i\in\N,\\
\omega^\rho_{i,0} \sim \exp(\frac12+\beta), &\text{ for }i\in\N, \\
\omega^\rho_{i,j} \sim \exp(1), &\text{ for }(i,j)\in {\cal B}.
\end{array}
\right.
\end{equation}
The LPP $L^{\beta,\alpha}$ has an explicit correlation kernel, see e.g.\ Theorem~3.1 in~\cite{BFO20}, which can be easily studied with similar computations as in the proof of Proposition~\ref{PropUpperBoundNotDiagonal}.

The second way to obtain \eqref{eqB.2} is to use the explicit formulas for the stationary LPP in Theorem~2.4 of~\cite{BFO20} and use the bounds for large $N$ obtained in Section~4 therein (a similar argument for the limiting object was presented in Appendix~D of~\cite{BFO21}).

By Lemma~\ref{lemCompFullSpaceHalfSpace} with $p=(N,N)$ and $q=(N+n,N-n)$, we get
\begin{equation}
L^{{\rm pp};1}(N+n,N-n)\geq L^{{\rm pp};1}(N,N)+ \left(L^{\square}(N+n,N-n)-L^{\square}(N,N)\right).
\end{equation}
Then
\begin{equation}
\begin{aligned}
\eqref{eqB.3}&\leq \Pb(L^{{\rm pp};1}(N,N)\leq 4N-\tfrac12 S 2^{4/3}N^{1/3})\\
&+\Pb(L^{\square}(N+n,N-n)-L^{\square}(N,N)\leq -\tfrac12 S 2^{4/3}N^{1/3}).
\end{aligned}
\end{equation}
Notice that $L^{{\rm pp};1}(N,N)\stackrel{(d)}{=}L^{{\rm pp};\infty}(N+1,N+1)$. Then, using \eqref{eqPPhalfSpacebounds} of Appendix~\ref{appLaguerreBeta} we get
\begin{equation}
 \Pb\left(L^{{\rm pp};1}(N,N)\leq 4N -\tfrac12 S 2^{4/3}N^{1/3}\right) \leq C e^{-c S^3}.
\end{equation}
For the last term, we use
\begin{equation}
\begin{aligned}
& \Pb(L^{\square}(N+n,N-n)-L^{\square}(N,N)\leq -\tfrac12 S 2^{4/3}N^{1/3}) \\
&\leq \Pb(L^{\square}(N+n,N-n)\leq 4N-\tfrac14 S 2^{4/3}N^{1/3})+\Pb(L^{\square}(N,N)\geq 4N -\tfrac14 S 2^{4/3}N^{1/3})\\
&\leq C e^{-c (S/4-M_1)^3} + C e^{-c S^{3/2}},
\end{aligned}
\end{equation}
where we used \eqref{eqPPfullSpacebounds} in the last step. Combining the bounds we obtained, the proof is completed.
\end{proof}

\begin{rem}
These bounds are very rough, but they are enough to ensure integrability of moments of our rescaled LPP problems. Also, notice that these simply derived bounds do not cover the regime studied in Appendix~\ref{appRHP}.
\end{rem}

\section{Lower tail bound for diagonal end-point}\label{appRHP}
In this section we derive a bound on a part of the lower tail of the distribution in the case where the end-point is on the diagonal. We will apply Riemann-Hilbert method of Deift-Zhou~\cite{DZ95}.

Consider half-space LPP with
\begin{equation}
\begin{aligned}
&\omega_{i,j}\sim \exp(1),& \textrm{ for }i>j,\\
&\omega_{i,i}\sim \exp(1/\tilde\rho),&\textrm{ for }i=j.
\end{aligned}
\end{equation}
Let $L_N(\tilde\rho)$ be the last passage percolation to the end-point $(N,N)$ and consider $1/\tilde\rho=\rho=\tfrac12+w 2^{-1/3}N^{-1/3}$, i.e., up to $o(1)$-terms which are irrelevant for the asymptotic behavior, we have
\begin{equation}
\tilde\rho=2-w 2^{5/3} N^{-1/3}.
\end{equation}

\begin{thm}\label{thmRHP}
Consider the scaling
\begin{equation}\label{eqAppRHP1}
x=4N+\xi 2^{4/3}N^{1/3}\quad\textrm{with}\quad \xi=\mu w^2,\quad \mu\in (0,4).
\end{equation}
There exists a constant $C$ such that for $\xi\in [0,o(N^{1/6})]$,
\begin{equation}
\Pb(L_N(\tilde\rho)\leq x) \leq C e^{\frac83 w^3-2 w \xi} e^{-\frac23\xi^{3/2}}= C e^{w^3 \left(\tfrac83-2\mu+\tfrac23\mu^{3/2}\right)}
\end{equation}
for all $N$ large enough. Notice that the function $\mu \mapsto \tfrac83-2\mu+\tfrac23\mu^{3/2}>0$ on $\mu\in (0,4)$ (and monotone decreasing).
\end{thm}
This result will be proven in the rest of this appendix. It is enough to prove it for $\xi\in [K,o(N^{1/6})]$ for some constant $K$, since by choosing the constant $C$ the estimate then holds also for $\xi\in [0,K]$.

\subsection{Distribution in terms of RHP}
Let us first explain how the distribution we are interested in is given in terms of the solution of a Riemann-Hilbert Problem (RHP). This has been explained to us by Jinho Baik.

We start with Theorem~1.3 of~\cite{Bai02}, which gives
\begin{equation}
\Pb(L_{N}(\tilde\rho)\leq x)=\frac12 \left\{\frac{a_N(x,\tilde\rho)-b_N(x,\tilde\rho)}{E_N(x)}+(a_N(x,\tilde\rho)+b_N(x,\tilde\rho)) E_N(x)\right\} F_N(x).
\end{equation}
The functions $E_N(x)$ and $F_N(x)$ do not depend on $\tilde\rho$ and they go to $1$ as $x\gg 4N$. In particular, one has the identities
\begin{equation}
a_N(x,2)=E_N(x)^2,\quad \Pb(L_{N}(2)\leq x)=E_N(x) F_N(x)
\end{equation}
and
\begin{equation}
\Pb(L^{\rm LUE}_{N}\leq x)=F_N(x)^2
\end{equation}
where $L^{\rm LUE}_N$ is the full-space LPP, which is given in terms of the Laguerre unitary ensembles with parameter $\alpha=0$.

The asymptotics of these two distributions are known. In particular there exists a constant $C$ such that for all $\xi\geq 0$ with $\xi=o(N^{1/6})$
\begin{equation}\label{eqApp0}
1- F_N(x)= \Or(e^{-\frac43 \xi^{3/2}}),\quad 1-E_N(x)=\Or(e^{-\frac23\xi^{3/2}})
\end{equation}
for all $N$ large enough. The derivation of the first expansion can be easily made using the determinantal structure of the \textrm{LUE} kernel (see e.g.\ the asymptotics in Section 7 of~\cite{FS05a} for $m=N$ and $d=0$ which gives $w=0$ therein). For the purpose of this appendix, the precise asymptotics of $F_N(x)$ is actually irrelevant. For the second, one can use the Pfaffian kernel, but it follows also from the asymptotics of the solution of the RHP problem described below, see Remark~\ref{remarkApp} below.

In (6.56)-(6.58) of~\cite{Bai02}, $a_N$ and $b_N$ are given in terms of solution of a Riemann-Hilbert problem $M^{(1)}$. They are $a_N(x,\tilde\rho)=M^{(1)}_{22}(2/\tilde\rho-1)$ and $b_N(x,\tilde\rho)=M^{(1)}_{12}(2/\tilde\rho-1)$. So we have
\begin{equation}
\Pb(L_{N}(\tilde\rho)\leq x)=\frac12 \left\{\frac{M_{22}^{(1)}(\omega)-M_{12}^{(1)}(\omega)}{E_N(x)}+(M_{22}^{(1)}(\omega)+M_{12}^{(1)}(\omega)) E_N(x)\right\} F_N(x),
\end{equation}
where
\begin{equation}
\omega=2/\tilde\rho - 1 = 2^{2/3} w N^{-1/3}.
\end{equation}

The $M^{(1)}$ is a solution of a RHP, which is a transformation of the the solution $M$ of another RHP as given by (6.53) of~\cite{Bai02}. To define it, let use define some domains and contours as follows. For any radius $r\in (0,1)$, define $\Gamma_1=\{z\in\C | |z-1|=r\}$ and $\Gamma_2=\{z\in\C | |z+1|=r\}$, where $\Gamma_1$ is anticlockwise oriented and $\Gamma_2$ is clockwise oriented. Let $\Gamma=\Gamma_1\cup\Gamma_2$ and define the regions $\Omega_1=\{z\in\C| |z-1|<r\}$, $\Omega_2=\{z\in\C| |z+1|<r\}$ and $\Omega_0$ the rest of $\C$ without $\Gamma$, see Figure~\ref{fig:RHPcontours}.
\begin{figure}
  \centering
   \includegraphics[height=3.5cm]{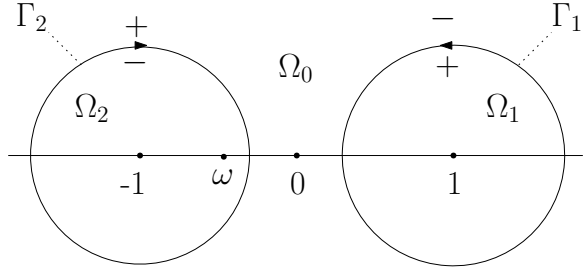}
   \caption{The contours $\Gamma_1$ and $\Gamma_2$ as well as the regions $\Omega_0,\Omega_1,\Omega_2$.}
 \label{fig:RHPcontours}
\end{figure}
Let \begin{equation}
\Phi(z;x)=\frac{(1+z)^N}{(1-z)^N}e^{-x z/2}
\end{equation} and define the jump matrix $v(z)$ by $v(z)=v_i(z)$ for $z\in\Gamma_i$, $i=1,2$, where
\begin{equation}
v_1(z)=\left(
       \begin{array}{cc}
         1 & -\Phi(z) \\
         0 & 1 \\
       \end{array}
     \right)\quad\textrm{and}\quad
v_2(z)=\left(
       \begin{array}{cc}
         1 & 0 \\
         \Phi(z)^{-1} & 1 \\
       \end{array}
     \right).
\end{equation}
Then the $2\times 2$ matrix $M$ is the solution of the RHP
\begin{equation}
\left\{
  \begin{array}{ll}
    M(z)\textrm{ is analytic in }z\in\C\setminus\Gamma,\\
    M_+(z)=M_-(z) v(z)\textrm{ for }z\in\Gamma,\\
    M(z)=\Id+\Or(z^{-1})\textrm{ as }z\to\infty.
  \end{array}
\right.
\end{equation}
Here with $M_+$ (resp. $M_-$) we mean the limit coming from the $+$ (resp. $-$) side of the contour $\Gamma$, where the $+$ side is the left hand side and $-$ is the right hand side of it (with respect of its orientation). For instance, for $\Gamma_1$, the $+$ side is $\Omega_1$ and the $-$ side is $\Omega_0$.

(6.53) of~\cite{Bai02} gives $M^{(1)}(z)=M(z)$ for $z\in\Omega_1$, $M^{(1)}(z)=M(z) v_1(z)$ for $z\in\Omega_0$ and $M^{(2)}(z)=M(z) v_2(z) v_1(z)$ for $z\in\Omega_2\setminus\{-1\}$.

For our purpose, we have $w<0$ which means that $\omega<0$. By choosing the radius $r\in (1+\omega,1)$ we have that $\omega\in \Omega_2$. Thus using the third expression in (6.53) of~\cite{Bai02} we get
\begin{equation}\label{eqApp1}
\Pb(L_{N}(\tilde\rho)\leq x)=\frac12 \left\{\frac{M_{11}(\omega)-M_{21}(\omega)}{E_N(x)}- (M_{11}(\omega)+M_{21}(\omega)) E_N(x)\right\} F_N(x)\Phi(\omega;x).
\end{equation}
A simple computation gives
\begin{equation}\label{eqApp2}
\Phi(\omega;x)= e^{\frac83 w^3-2 w \xi}(1+\Or(w^5 n^{-2/3})).
\end{equation}
Below we will show that for all $\xi=o(N^{1/6})$ but large enough,
\begin{equation}\label{eqApp3}
\begin{aligned}
M_{11}(\omega)&=1+\Or(e^{-\frac23\xi^{3/2}}),\\
M_{21}(\omega)&=\Or(e^{-\frac23\xi^{3/2}}).
\end{aligned}
\end{equation}
Then, plugging \eqref{eqApp0}, \eqref{eqApp2}, and \eqref{eqApp3} into \eqref{eqApp1} we obtain the statement of Theorem~\ref{thmRHP}.

\subsection{Asymptotics of the RHP problem}
For the asymptotic, we consider $w<0$ and $x$ scaled as in \eqref{eqAppRHP1}. We choose the contours $\Gamma_1,\Gamma_2$ to have radius
\begin{equation}
 r=1-\tfrac12 \e\quad\textrm{with}\quad \e=\sqrt{\xi} 2^{2/3}N^{-1/3}.
\end{equation}
With this choice we have $\omega=2^{2/3}w N^{-1/3}< -\tfrac12 \e=-\sqrt{\xi/4} 2^{2/3}N^{-1/3}\in\Omega_2$ for any $\xi\in (0,4 w^2)$ as required.

Let us now consider the scaling of the variable $z:=\e Z$. Then the new RHP is on the contours
\begin{equation}
\tilde \Gamma_1=\{z\in\C | |Z-\e^{-1}|=\e^{-1}-1/2\},\quad \tilde \Gamma_2=\{z\in\C | |Z+\e^{-1}|=\e^{-1}-1/2\}
\end{equation}
and define $\tilde M(Z)=M(\e Z)$. In particular, the quantity of interest is given by
\begin{equation}
M(\omega)=\tilde M(w/\sqrt{\xi}).
\end{equation}
The new jump matrices are given by $\tilde v(z)=\tilde v_i(z)$ for $z\in\tilde \Gamma_i$, $i=1,2$,
\begin{equation}
 \tilde v_1(Z)=\left(
               \begin{array}{cc}
                 1 & -\varphi(Z) \\
                 0 & 1 \\
               \end{array}
             \right)\quad\textrm{and}\quad
 \tilde v_2(Z)=\left(
               \begin{array}{cc}
                 1 & 0 \\
                 \varphi(Z)^{-1} & 1 \\
               \end{array}
             \right)
\end{equation}
with
\begin{equation}
 \varphi(Z)=\Phi(\e Z)= \left(\frac{1+\e Z}{1-\e Z}\right)^N e^{-\tfrac12\e Z (4N+\xi 2^{4/3}N^{1/3})}.
\end{equation}
Using the identity $\xi 2^{4/3}N^{1/3}=\e^2 N$ we obtain
\begin{equation}
 \varphi(Z)= e^{N f(Z)},\quad f(Z)=\ln(1+\e Z)-\ln(1-\e Z)-\e Z(2+\e^2/2).
\end{equation}

\subsubsection*{Steep descent property of $\varphi(Z)$ on $\tilde\Gamma_1$ and $\varphi(Z)^{-1}$ on $\tilde\Gamma_2$}
Since $\varphi(Z)^{-1}=\varphi(-Z)$ and (up to the orientation) $\tilde\Gamma_2=-\tilde\Gamma_1$, it is enough to consider $\varphi(Z)$ on $\tilde\Gamma_1$.

\begin{lem}\label{lemRHPSteepdesc}
Let us parameterize $\tilde\Gamma_1$ by $Z=\e^{-1}-(\e^{-1}-\tfrac12) e^{\I\theta}$. Then, for all $N$ large enough,
\begin{equation}
 |\varphi(Z)|\leq e^{-\frac23 \xi^{3/2}(1+\Or(\e))} e^{-c\sqrt{\xi}N^{2/3}\theta^2},
\end{equation}
with $c=2^{2/3} \,4 /\pi^2$.
\end{lem}
\begin{proof}
We need to find a bound for $|\varphi(Z)|= e^{\Re(f(Z))}$. Thus we are interested in the critical points of $f$. We have
\begin{equation}
\frac{d}{dZ} f(Z)=-\frac{\epsilon ^3 \left(Z^2 \epsilon ^2+4 Z^2-1\right)}{2 (Z \epsilon -1) (Z \epsilon +1)} =0\quad \leftrightarrow\quad Z=Z_\pm=\pm \frac{1}{\sqrt{4+\e^2}} \pm\frac12+o(\e).
\end{equation}
This explains the choice of the radius of our contours. Indeed, $\tilde\Gamma_1$ passes almost at the critical point $Z_+$. Then
\begin{equation}
\Re(f(Z))=\frac12\ln|1+\e Z|^2-\frac12 \ln|1-\e Z|^2-\Re(Z)\e(1+\e^2/2).
\end{equation}
Using $|1+\e Z|^2=4+(1-\e/2)^2-4 (1-\e/2)\cos(\theta)$, $|1-\e Z|^2=(1-\e/2)^2$ and $\Re(Z)\e=1-(1-\e/2)\cos(\theta)$ we get
\begin{equation}
  \frac{d \Re(f(Z))}{d\theta}=-2\sin(\theta)(1-\e/2)\left[1+\frac14\e^2-\frac{1}{|1+\e Z|^2}\right].
\end{equation}
When $\theta$ increases on $[0,\pi]$, the term in the squared brackets also increases with
\begin{equation}
 1+\frac14\e^2-\frac{1}{|1+\e Z|^2} \geq 1+\frac14\e^2-\frac{1}{(1+\e/2)^2}=\e(1+\Or(\e)).
\end{equation}
Thus we get
\begin{equation}
  \frac{d \Re(f(Z))}{d\theta}\leq -2\sin(\theta)\e(1+\Or(\e))\leq -\e\sin(\theta)
\end{equation}
for all $\e>0$ small enough. From this and $1-\cos(\theta)\geq \theta^2 2/\pi^2$, it follows that for $\theta\in [0,\pi]$,
\begin{equation}
\Re(f(Z))\leq \Re(f(1/2))-\e\theta^2 2/\pi^2.
\end{equation}
A simple computation leads to
\begin{equation}
 N\Re(f(1/2))=-\frac16N\e^3(1+\Or(\e))=-\frac23\xi^{3/2}(1+\Or(\e)).
\end{equation}
Recalling that $\e=\sqrt{\xi} 2^{2/3}N^{-1/3}$ we have the claimed result.
\end{proof}

As a consequence of Lemma~\ref{lemRHPSteepdesc} we obtain the following estimates on the jump matrix.
\begin{cor}\label{CorRHP}
For all $\xi=o(N^{1/6})$, there exist a constant $C$ so that
\begin{equation}
 \begin{aligned}
 \|\tilde v-\Id\|_{L^\infty(\tilde\Gamma)}&\leq C e^{-\frac23 \xi^{3/2}},\\
\|\tilde v-\Id\|_{L^2(\tilde\Gamma)}&\leq C e^{-\frac23\xi^{3/2}}/\xi^{3/8},\\
\|\tilde v-\Id\|_{L^1(\tilde\Gamma)}&\leq C e^{-\frac23\xi^{3/2}}/\xi^{3/4}.
\end{aligned}
\end{equation}
for all $N$ large enough.
\end{cor}
\begin{proof}
The bound follows (up to a different constant $C$) from the same bounds on the path $\tilde\Gamma_1$. The bound on $\|\tilde v-\Id\|_{L^\infty(\tilde\Gamma)}$ is a direct consequence of Lemma~\ref{lemRHPSteepdesc} and the fact that $\xi^{3/2}\e=\Or(\xi^2 N^{-1/3})=o(1)$ for $\xi=o(N^{1/6})$. For the $L^2$ bound, we need just to compute a Gaussian integral, namely
\begin{equation}
 \|\tilde v-\Id\|^2_{L^2(\tilde\Gamma_1)}\leq \int_{\tilde\Gamma_1} |dz| |\varphi(Z)|^2\leq C e^{-\frac43 \xi^{3/2}}\int_{-\pi}^\pi \frac{d\theta}{\e} e^{-2c\sqrt{\xi}N^{2/3}\theta^2}\leq  C' e^{-\frac43 \xi^{3/2}}/\xi^{3/4}.
\end{equation}
Similarly for the $L^1$ bound.
\end{proof}

\subsubsection*{Solution of the RHP}
Using the notations of Deift-Zhou seminal paper~\cite{DZ95}, we decompose the jump matrix as $\tilde v=\Id+ \omega_+$ (we have $ \omega_-=0$), so that by Corollary~\ref{CorRHP} $\omega_+=\tilde v-\Id\in L^2(\tilde\Gamma)$. Then we have
\begin{equation}\label{eqApp4}
\tilde M(Z)=\Id+\int_{\tilde\Gamma} \frac{d\zeta}{2\pi\I} \frac{\omega_+(\zeta)}{\zeta-Z} + \int_{\tilde\Gamma}\frac{d\zeta}{2\pi\I}\frac{[(1-{\cal C}_{\omega_+})^{-1} {\cal C}_{\omega_+} \Id](\zeta)\omega_+(\zeta)}{\zeta-Z}.
\end{equation}
We need still to define the operator ${\cal C}_{\omega_+}$. This is given in terms of the Cauchy operator $\cal C$
\begin{equation}
{\cal C}f(z)=\frac{1}{2\pi\I}\int_{\tilde\Gamma}\frac{f(\zeta)}{z-\zeta}d\zeta.
\end{equation}
Then define ${\cal C}_- f(z)=\lim_{z_-\to z}{\cal C}f(z_-)$ where the limit is taken from the $-$ side of $\tilde\Gamma$. The operator ${\cal C}_-$ is well-defined in $L^2(\tilde\Gamma)$ and has a finite $L^2$-norm. Then we have (in the case of the decomposition $\tilde v=\Id+\omega_+$),
\begin{equation}
{\cal C}_{\omega_+} f={\cal C}_-(f\omega_+).
\end{equation}

From \eqref{eqApp4} we can bound on $\tilde M(Z)-\Id$ as
\begin{equation}
\|\tilde M(Z)-\Id\|\leq \frac{1}{{\rm dist}(\tilde\Gamma,Z)} \left( \|\omega_+\|_{L^2(\tilde\Gamma)}+
 \|(\Id-{\cal C}_{\omega_+})^{-1}\|_{L^2(\tilde\Gamma)\to L^2(\tilde\Gamma)} \, \|{\cal C}_{\omega_+}\Id\|_{L^2(\tilde\Gamma)}\, \|\omega_+\|_{L^2(\tilde\Gamma)}\right).
\end{equation}
We have $\|{\cal C}_{\omega_+}\|_{L^2(\tilde\Gamma)\to L^2(\tilde\Gamma)}\leq \|{\cal C}_-\|_{L^2(\tilde\Gamma)\to L^2(\tilde\Gamma)} \|\omega_+\|_{L^\infty(\tilde\Gamma)}$ and thus, for $\xi$ large enough, $\|{\cal C}_{\omega_+}\|_{L^2(\tilde\Gamma)\to L^2(\tilde\Gamma)}\leq 1/2$, which implies $\|(\Id-{\cal C}_{\omega_+})^{-1}\|_{L^2(\tilde\Gamma)\to L^2(\tilde\Gamma)}\leq 2$. Furthermore, since ${\cal C}_{\omega_+} \Id = {\cal C}_- \omega_+$, we get
$\|{\cal C}_{\omega_+} \Id\|_{L^2(\tilde\Gamma)}\leq \|{\cal C}_-\|_{L^2(\tilde\Gamma)\to L^2(\tilde\Gamma)}  \|\omega_+\|_{L^2(\tilde\Gamma)}$.

Using the bounds of Corollary~\ref{CorRHP} we get that for all $\xi$ large enough and $\xi=o(N^{1/6})$,
\begin{equation}
\|\tilde M(Z)-\Id\|\leq \frac{1}{{\rm dist}(\tilde\Gamma,Z)} C e^{-\frac23\xi^{3/2}} \xi^{-3/8}.
\end{equation}

In our setting, $Z=w/\sqrt{\xi}$, so that ${\rm dist}(\tilde\Gamma,Z)\geq |w/\sqrt{\xi}-1/2|>0$. Thus the estimates \eqref{eqApp3} are proven.

\begin{remark}\label{remarkApp}
Using the above computations, but evaluating the solution of the RHP at $0$ one obtains a bound for $E_N(x)$ since $E_N(x)^2=a_N(x;2)=M_{22}(0)-M_{21}(0) \Phi(0)$. As $\Phi(0)=1$ and $\tilde M(0)=M(0)$ we get $E_N(x)^2=\tilde M_{22}(0)-\tilde M_{21}(0)=1+\Or(e^{-\frac23\xi^{3/2}} \xi^{-3/8})$.
\end{remark}


\end{document}